\documentclass{article}

\usepackage[ruled]{algorithm2e}
\usepackage{amsmath, amsthm, amssymb}
\usepackage{fullpage}
\usepackage{verbatim}
\usepackage{xcolor}
\usepackage{graphicx}
\usepackage{framed}
\usepackage[numbers]{natbib}

\newtheorem{theorem}{Theorem}
\newtheorem{lemma}{Lemma}

\newcommand{\figwidth}{0.405\linewidth}

\newcommand{\cS}{\mathcal{S}}
\newcommand{\cA}{\mathcal{A}}
\newcommand{\cH}{\mathcal{H}}
\newcommand{\up}{u_P}
\newcommand{\ua}{u_A}
\newcommand{\vp}{v^P}
\newcommand{\va}{v^A}
\newcommand{\trans}{P}
\newcommand{\etrans}{P^E}
\newcommand{\last}{\mathrm{last}}
\newcommand{\eps}{\varepsilon}
\newcommand{\optstatmech}{\mathsf{OptStatMech}}
\newcommand{\rand}{\mathrm{rand}}

\DeclareMathOperator*{\argmax}{\mathrm{argmax}}

\title{Automated Dynamic Mechanism Design}

\date{}

\author{
    Hanrui Zhang \\
    Duke University \\
    \texttt{hrzhang@cs.duke.edu}
\and
    Vincent Conitzer \\
    Duke University \\
    \texttt{conitzer@cs.duke.edu}
}

\begin{document}

\maketitle

\begin{abstract}
    We study Bayesian automated mechanism design in unstructured dynamic environments, where a principal repeatedly interacts with an agent, and takes actions based on the strategic agent's report of the current state of the world.
    Both the principal and the agent can have arbitrary and potentially different valuations for the actions taken, possibly also depending on the actual state of the world.
    Moreover, at any time, the state of the world may evolve arbitrarily depending on the action taken by the principal.
    The goal is to compute an optimal mechanism which maximizes the principal's utility in the face of the self-interested strategic agent.

    We give an efficient algorithm for computing optimal mechanisms, with or without payments, under different individual-rationality constraints, when the time horizon is constant.
    Our algorithm is based on a sophisticated linear program formulation, which can be customized in various ways to accommodate richer constraints.
    For environments with large time horizons, we show that the principal's optimal utility is hard to approximate within a certain constant factor, complementing our algorithmic result.
    We further consider a special case of the problem where the agent is myopic, and give a refined efficient algorithm whose time complexity scales linearly in the time horizon.
    Moreover, we show that memoryless mechanisms, which are without loss of generality optimal in Markov decision processes without strategic behavior, do not provide a good solution for our problem, in terms of both optimality and computational tractability.
    These results paint a relatively complete picture for automated dynamic mechanism design in unstructured environments.
    Finally, we present experimental results where our algorithms are applied to synthetic dynamic environments with different characteristics, which not only serve as a proof of concept for our algorithms, but also exhibit intriguing phenomena in dynamic mechanism design.
\end{abstract}

\section{Introduction}

Consider the following scenario.
A company assembles an internal research group to develop key technologies to be used in the company's next-generation product in 5 years.
The more progress the group makes, the more successful the product is likely to be.
Since research progress is hard to monitor, the company manages the group based on its annual reports.
At the beginning of each year, the group submits a report, summarizing its progress in the preceding year, as well as its needs for the current year.
Taking into consideration this report (and possibly also reports from previous years), the company then decides the compensation level and the headcount of the group in the current year.
Moreover, after the product launches, the company may also pay a bonus to members of the group, depending on how successful the product is.

For simplicity, suppose an annual report consists of two items: research progress (satisfactory/unsatisfactory), and need to expand (no request/request for an intern/request for a full-time employee).
The company's goal is to encourage and facilitate research progress while keeping the expenses reasonable.
So, a natural managing strategy is to increase (resp.\ decrease) the compensation level when the reported research progress is satisfactory (resp.\ unsatisfactory), and to allow the group to expand only when necessary, i.e., when the reported research progress is unsatisfactory.
However, the research group may have a different goal than the company's.
Suppose members of the group do not care about the success of the product {\em per se}.
Instead, their primary goal is to maximize the total compensation received from the company, and for this reason, they may be incentivized to {\em misreport} the situation.
In other words, the company faces a {\em dynamic mechanism design} problem, where the {\em principal} (i.e., the company) needs to implement (and commit to) a {\em mechanism} (i.e., a managing strategy) that achieves its goal through {\em repeated} interactions, in the presence of strategic behavior of the {\em agent} (i.e., the research group).

Indeed this problem is nontrivial.
For example, if the company implements the above strategy, then the group will report satisfactory progress regardless of the actual situation, which maximizes the group's total compensation over the 5 years, but also causes greater expenses for the company and jeopardizes the success of the product.
To counter this, 
the company may additionally promise a significant bonus contingent on the success of the product.
This creates incentives for the group to make more progress, and discourages overreporting the progress, because the group is not allowed to expand when the reported progress is satisfactory.  That is, if actual progress is unsatisfactory, this introduces an incentive to report this truthfully so that the group may expand.
However, this also runs the risk of encouraging the group to report unsatisfactory progress in order to expand even if actual progress is satisfactory, because more members always make more progress, which leads to a higher (chance of) bonus, whereas the cost of expanding is paid by the company and therefore irrelevant to the group.

One may try to fix this by introducing more rules, possibly replacing existing ones.
For example, the company may allow the group to recruit an intern, but not a full-time employee, when the reported progress is unsatisfactory.
Then, in the next year, if the reported progress improves, the company allows the group to make a return offer to the intern as a full-time employee.
Or alternatively, the company may unconditionally allow the group to recruit interns (which are less costly), but never full-time employees.
In addition to the above, the company could also temporarily decrease the compensation level when a new member joins, and later adjust the compensation based on how the reported progress improves.
While all these ad hoc rules make intuitive sense, it is not immediately clear which (combinations of) rules are better, how to optimize parameters of these rules (e.g., the number of new members allowed per year and the amount by which the compensation is adjusted), or whether there is a better set of rules that look totally different.

As demonstrated by the foregoing discussion, in general, the problem of finding an optimal mechanism in {\em unstructured} dynamic environments, such as the above example, turns out to be extremely rich and challenging.
In such environments, the actions of the principal may go beyond the allocation of items to the agent, and affect the state of the world in arbitrary ways.
Moreover, both the principal and the agent may have arbitrary valuations for these actions, which also depend on the current state of the world.
In economic theory, the {\em characterize-and-solve} approach \cite{myerson1981optimal,courty2000sequential,pavan2014dynamic} to mechanism design has achieved spectacular success in both static and dynamic environments, by exploiting structure of the environment to construct a characterization of optimal mechanisms, often leading to closed-form or computationally tractable solutions.
However, since the environments under consideration here are loosely structured at best, the classical characterize-and-solve approach does not seem particularly suited.
When disregarding the agent's incentives, one could treat the problem of finding an optimal strategy as a {\em planning} problem, which is known to be solvable efficiently \cite{bellman1957markovian,howard1960dynamic,puterman1978modified}.
However, as discussed above, the agent's strategic behavior can ruin the performance of such a strategy.
From a computational perspective, while numerous methods for {\em automated mechanism design}, which efficiently compute optimal mechanisms without heavily exploiting structures of the environment, have been proposed \cite{conitzer2002complexity,conitzer2004self,conitzer2006computing}, all existing methods work only for static environments with one-time interactions, and it is not immediately clear how to generalize these methods to dynamic environments.
All this brings us to the following question:

\medskip
\begin{framed}
    \centering
    \em Can we efficiently compute optimal mechanisms in unstructured dynamic environments?
\end{framed}

\subsection{Our Results}

In this paper, we study the problem of computing optimal mechanisms in {\em single-agent}, {\em discrete-time} dynamic environments with a {\em finite time horizon}, without any further structural assumptions.
Our main results (presented in Section~\ref{sec:general}) can be summarized as follows:
\begin{itemize}
    \item {\bf Efficient algorithm}: when the time horizon is fixed, there is a polynomial-time algorithm for computing optimal mechanisms, with or without payments, that maximize the principal's utility facing a strategic agent.
    \item {\bf Inapproximability}: when the time horizon can be large, it is $\mathsf{NP}$-hard to approximate the principal's optimal utility within a factor of $(7/8 + \eps)$ for any $\eps > 0$.
\end{itemize}
To the best of our knowledge, our algorithm for constant time horizons is the first that efficiently computes optimal mechanisms in unstructured dynamic environments.
The fact that our algorithm cannot scale beyond constant time horizons is by no means surprising: optimal dynamic mechanisms generally depend on the entire history, and as a result, the straightforward description of such a mechanism is exponentially large in the time horizon.
Our inapproximability result further rules out the possibility of computing succinct representations of approximately optimal mechanisms that can be efficiently evaluated.
These results together paint a complete picture of the computational complexity of dynamic mechanism design in unstructured environments.

In Section~\ref{sec:myopic}, we zoom into a special case of the problem where the agent is {\em myopic}, i.e., where the agent cares only about immediate value when making decisions.
This is still practically meaningful, since it is commonly assumed and observed that the principal is often much more patient than the agent in dynamic environments.
(This could also correspond to the agent really being a sequence of short-lived agents; for example, there may be high turnover in the research group in the example above, where each researcher is there only for one period.)
We show that in such cases, without loss of generality, optimal mechanisms admit succinct representations, i.e., they depend only on the current state and time, the previous state, and the previous action.
Based on this characterization, we provide an improved algorithm for finding optimal mechanisms in the face of a myopic agent, whose time complexity depends {\em linearly} on the time horizon.
As a result, this algorithm scales well in dynamic environments with long time horizons, which is in sharp contrast to the general case where long time horizons lead to inapproximability.

As discussed above, without strategic behavior, our problem degenerates to the problem of planning in (finite episodic) Markov Decision Processes (MDPs).
It is known that in MDPs, optimal strategies are without loss of generality {\em memoryless}: they depend only on the current time and state.
To this end, one may wonder if memoryless mechanisms are also (approximately) optimal and/or easier to compute in dynamic environments with strategic behavior.
In Section~\ref{sec:memoryless}, we give a negative answer to the above question, by showing that (1) the principal's optimal utility achieved by memoryless mechanisms can be arbitrarily worse than that achieved by general dynamic mechanisms, and (2) it is $\mathsf{NP}$-hard to approximate the principal's optimal utility achieved by memoryless mechanisms within a factor of $(7/8 + \eps)$ for any $\eps > 0$.
In other words, memoryless mechanisms do not provide a good solution for our problem, in terms of both optimality and computational tractability.

Finally, in Section~\ref{sec:exp}, we apply our algorithms to synthetic dynamic environments with different characteristics, in order to provide a proof of concept for the methods we propose, as well as to explore various phenomena in dynamic environments and their implications for (automated) dynamic mechanism design.
Below are some of our key findings:
\begin{itemize}
    \item As in static environments, taking into consideration the agent's incentives in dynamic environments can greatly improve the principal's utility.
    \item In dynamic environments, optimal mechanisms are remarkably robust to misaligned interests between the principal and the agent, whereas the performance of na\"ive mechanisms (which disregard the agent's incentives) degrades much faster.
    \item Even when the principal's and the agent's valuations are perfectly aligned, an agent acting myopically can still considerably hurt the principal's utility in na\"ive mechanisms, but this can be largely corrected by using mechanisms that are optimal in the face of a myopic agent.
    \item As one would expect, patient agents are easier to cooperate with, and myopic agents are easier to exploit; however, even when the principal’s and the agent’s valuations are negatively correlated, it is possible to find a middle ground where cooperation is more beneficial than exploitation in the long run.
\end{itemize}

\subsection{Further Related Work}

\paragraph{Dynamic mechanism design.}
The problem we study can be situated in the broad area of dynamic mechanism design, and below we discuss some representative related work.
For a more comprehensive exposition, see, e.g., the survey by \citet{pavan2017dynamic} and the one by \citet{bergemann2019dynamic}.
In the context of efficient (i.e., welfare-maximizing) mechanisms, \citet{bergemann2010dynamic} propose the dynamic pivot mechanism, which generalizes the VCG mechanism in static environments, and \citet{athey2013efficient} propose the team mechanism, which focuses on budget-balancedness.
As for optimal (i.e., revenue-maximizing) mechanisms, which are more closely related to our results, following earlier work \cite{baron1984regulation,courty2000sequential,esHo2007optimal}, \citet{pavan2014dynamic} generalize the classical characterization by \citet{myerson1981optimal} into dynamic environments, unifying previous results with continuous type spaces.
\citet{ashlagi2016sequential} study ex-post individual-rational dynamic mechanisms for repeated auctions, and give an efficient $(1 - \eps)$-approximation to the optimal revenue for a single agent with independent valuations across items.
\citet{mirrokni2020non} study non-clairvoyant dynamic mechanism design, where future distributional knowledge is unavailable to the principal.
All these results for optimal mechanisms follow the characterize-and-solve approach, which is quite different from the computational approach that we take.

Particularly related to our results is the work by \citet{papadimitriou2016complexity}, who study a setting where one item is sold at each time, and agents' valuations can be correlated across items.
They show that designing an optimal deterministic mechanism is computationally hard even when there is only one agent and two items (thereby ruling out the possibility of efficiently computing optimal deterministic mechanisms in our model, which is more general).
And moreover, they give a polynomial-size linear program formulation for optimal randomized mechanisms for independent agents when the number of agents and the time horizon are both constant.
Restricted to a single agent, their LP formulation can be viewed as a special case of our main result: they focus on revenue maximization with a single item to be allocated at each time, in a model where the principal's actions cannot affect the future valuations of the agent; on the other hand, we allow the principal to care about actions as well as revenue, with actions being general and unstructured (as opposed to allocation/no allocation), where the future state of the world can depend arbitrarily on the principal's actions as well as the current state.

Another related line of work is devoted to studying the design of repeated allocation mechanisms without money, motivated for example by allocating shared computing resources over time \cite{guo2009competitive,freeman2018dynamic,balseiro2019multiagent,gorokh2019monetary}.
When there is no money, repeated allocation allows one to better take current preferences for the items into account, because one can ``pay'' for one's current allocation with one's future allocations.
Indeed, a common theme of this line of work is to introduce artificial currencies or to approximate mechanisms {\em with} money via the use of future allocations.  The algorithms we present here can be used to find optimal mechanisms without money directly.


\paragraph{Automated Mechanism Design.}
There is a rich body of research regarding automated mechanism design (AMD) in (essentially) static environments.
\citet{conitzer2002complexity,conitzer2004self} initiated the study of automated mechanism design.
They consider various specific static setups, and show that computing optimal deterministic mechanisms, even with a single agent, is often $\mathsf{NP}$-hard (which also rules out the possibility of efficiently computing optimal deterministic mechanisms in our model, since the 1-period case is a special case), while computing optimal randomized mechanisms is often tractable.
Conceptually related to our model, \citet{hajiaghayi2007automated} consider a model where agents enter and leave the mechanism online (but still have one-time interactions with the mechanism), and provide an algorithm for computing mechanisms that are competitive against the optimal allocation from hindsight. \citet{sandholm2007automated} study automated design of multistage mechanisms, but these are not for dynamic settings; instead, the motivation is to implement static mechanisms using multiple rounds of queries in order to minimize the communication cost.
\citet{sandholm2015automated} study automated design of combinatorial auction mechanisms, and \citet{balcan2016sample,balcan2018general} study the sample complexity thereof.
\citet{kephart2015complexity,kephart2016revelation} and \citet{zhang2021automated} study AMD with partial verification and/or reporting costs.
More recently, various methods have been proposed for automated mechanism design via machine learning \cite{narasimhan2016automated}, and in particular, deep learning \cite{feng2018deep,dutting2019optimal,shen2019automated,rahme2020permutation}.
All these results are essentially for static environments, whereas in this paper, we focus solely on AMD in dynamic environments.

\paragraph{Equilibrium computation.}
Our main result can be viewed as an efficient algorithm for computing Stackelberg equilibria in a special class of extensive-form games.
Equilibrium computation is quite well understood in normal-form games, where there are polynomial-time algorithms for computing a Stackelberg equilibrium \cite{conitzer2006computing}, or a Nash equilibrium when the game is zero-sum (see, e.g., \cite{tardos2007basic}), in two-player games, whereas finding a Nash equilibrium in general-sum two-player games is already $\mathsf{PPAD}$-complete \cite{daskalakis2009complexity,chen2006settling}.
For extensive-form games, \citet{von1996efficient} and \citet{koller1996efficient} propose the sequence-form representation, which leads to an efficient algorithm for finding a Nash equilibrium (which is also a Stackelberg equilibrium) in two-player zero-sum games.
However, as shown by \citet{letchford2010computing}, computing a Stackelberg equilibrium in two-player general-sum extensive-form games is $\mathsf{NP}$-hard in general.
Polynomial-time (exact or $(1 - \eps)$-approximation) algorithms are known only for highly restrictive cases, e.g., in perfect-information settings \cite{bovsansky2017computation}, or when the follower is a finite state machine with limited memory \cite{vcerny2020finite} (although practically scalable algorithms exist for more general settings \cite{bosansky2015sequence,cermak2016using,vcerny2018incremental,kroer2018robust}).
Our results push the boundary of tractability of Stackelberg equilibrium in extensive-form games, by enabling efficient computation in a nontrivial class of {\em general-sum} extensive-form games with {\em imperfect information}.


\section{Preliminaries}

\paragraph{Dynamic environments.}
Throughout this paper, we consider single-agent, discrete-time environments with a finite time horizon.
Below, we give a general definition of such a dynamic environment.
Let $T$ be the time horizon, $\cS$ be the state space, and $\cA$ be the action space.
The agent observes the state, but the principal controls the action that is taken.
For each $t \in [T]$, let $\vp_t: \cS \times \cA \to \mathbb{R}$ be the principal's valuation function, where for each state $s \in \cS$ and action $a \in \cA$, $\vp_t(s, a)$ is the value of the principal when playing action $a$ in state $s$, at time $t$; similarly, let $\va_t: \cS \times \cA \to \mathbb{R}$ be the agent's value function.
Let $\trans_0 \in \Delta(\cS)$ be the initial distribution over states, and for each $s \in \cS$, denote by $\trans_0(s)$ the probability that the initial state is $s$.
Moreover, for each $t \in [T]$, let $\trans_t: \cS \times \cA \to \Delta(\cS)$ be the transition operator, which maps a state-action pair $(s, a)$ at time $t$ to the distribution of the next state at time $t + 1$, $\trans_t(s, a) \in \Delta(\cS)$.
We denote by $\trans_t(s, a, s')$ the probability that the next state is $s'$ when playing action $a$ in state $s$ at time $t \in [T]$.
For notational simplicity, let $\trans_0(s, a, s') = \trans_0(s')$ for all $s, s' \in \cS$ and $a \in \cA$.
(Note that the first {\em actual} action is taken at $t=1$ --- not $t=0$ --- possibly based on the state at $t=1$.)

\paragraph{Histories.}
A $t$-step history is a sequence of states and actions $(s_1, a_1, s_2, \dots, a_{t - 1}, s_t, a_t)$, where for each $i \in [t]$, it is the case that $s_i \in \cS$ and $a_i \in \cA$.
For each $t \in [T]$, let $\cH_t$ be the set of all possible $t$-step histories, i.e.,
\[
    \cH_t = \{(s_1, a_1, \dots, s_t, a_t) \mid s_i \in \cS, a_i \in \cA \text{ for all } i \in [t]\}.
\]
For each $h = (s_1, a_1, \dots, s_t, a_t) \in \cH_t$, let $|h| = t$, and moreover, for any $s_{t + 1} \in \cS$, $a_{t + 1} \in \cA$, let $h + (s_{t + 1}, a_{t + 1}) = (s_1, a_1, \dots, s_{t + 1}, a_{t + 1})$.
Let $\cH_0 = \{\emptyset\}$, where $\emptyset$ corresponds to the empty history with $|\emptyset| = 0$.
Let $\cH = \cH_0 \cup \bigcup_{t \in [T - 1]} \cH_t$ be the set of all possible histories of length at most $T - 1$ in the dynamic environment.
Note that $\cH$ does not contain histories of length $T$.

\paragraph{Dynamic mechanisms.}
Dynamic mechanisms are more powerful than static ones, in that they may depend on the {\em entire history}, rather than only the current state.
A (randomized) dynamic mechanism $M = (\pi, p)$ consists of an action policy $\pi$ and a payment function $p$.
The action policy $\pi: \cH \times \cS \to \Delta(\cA)$ maps each history $h \in \cH$, extended with the reported current state $s \in \cS$, to a distribution over actions $\pi(h, s) \in \Delta(\cA)$.
We denote by $\pi(h, s, a)$ the probability that the action taken by the mechanism is $a$ for $(h, s)$.
The payment function $p: \cH \times \cS \to \mathbb{R}$ maps the extended history $(h, s)$ to a real number, i.e., the payment, made from the agent to the principal (but it can be negative).
We remark that in principle, one can absorb payments into the action space.
However, doing so would make the action space uncountable, introducing subtleties into the computational problem (which is the main focus of this paper).
Here, we keep payments separate and explicit to avoid such issues.
Also, our algorithm allows linear constraints on feasible payments, including but not limited to: nonnegative payments, no payments, etc.
See Section~\ref{sec:short-horizon} for more details.

\paragraph{Utilities.}
Fixing a mechanism $M = (\pi, p)$, we can then define the onward utility of the principal and the agent.
Let $\up^M: \cH \times \cS \to \mathbb{R}$ be the principal's onward utility function under mechanism $M$, defined inductively such that
\[
    \up^M(h, s) = \sum_a \pi(h, s, a) \cdot \left(\vp_{|h| + 1}(s, a) + \sum_{s'} \trans_{|h| + 1}(s, a, s') \cdot \up^M(h + (s, a), s')\right) + p(h, s),
\]
with the boundary condition that $\up^M(h, s) = 0$ for all $h \in \cH_T$ and $s \in \cS$.
Here, all summations are over the entire state/action space.
Let $\up^M(\emptyset)$ be the overall utility of the principal, i.e.,
\[
    \up^M(\emptyset) = \sum_s \trans_0(s) \cdot \up^M(\emptyset, s).
\]
Similarly, let $\ua^\pi: \cH \times \cS \to \mathbb{R}$ be the agent's onward utility function under mechanism $M$, defined such that
\[
    \ua^M(h, s) = \sum_a \pi(h, s, a) \cdot \left(\va_{|h| + 1}(s, a) + \sum_{s'} \trans_{|h| + 1}(s, a, s') \cdot \ua^M(h + (s, a), s')\right) - p(h, s),
\]
where $\ua^M(h, s) = 0$ for all $h \in \cH_T$ and $s \in \cS$.
And let $\ua^M(\emptyset)$ be the overall utility of the agent, i.e.,
\[
    \ua^M(\emptyset) = \sum_s \trans_0(s) \cdot \ua^M(\emptyset, s).
\]
We remark that while the above definition assumes that the principal cares about payments as much as the agent does, in fact, our algorithm allows for the principal to care about payments in an arbitrary linear way (including possibly not at all).
See Section~\ref{sec:short-horizon} for a detailed discussion.

\paragraph{Incentive-compatible mechanisms.}
We say a mechanism $M$ is incentive-compatible (IC) if the agent can never achieve a higher overall utility by misreporting the state, even in  sophisticated ways.
Formally, a reporting strategy $r: \cH \times \cS \to \cS$ maps each history $h$ extended with the current state $s$ to a reported state $s'$, which is possibly different from $s$.
This reporting strategy induces a reported history $r(h) = (s_1', a_1, \dots, s_t', a_t)$ for each actual history $h = (s_1, a_1, \dots, s_t, a_t)$, where for each $i \in [t]$,
\[
    s_i' = r((s_1, a_1, \dots, s_{i - 1}, a_{i - 1}), s_i).
\]
Note that we abuse notation here: in particular, $r(h, s)$ denotes a reported state, whereas $r(h)$ denotes a reported history.
And without loss of generality, we only consider deterministic reporting strategies.
Given a mechanism $M$ and a reporting strategy $r$, we can define the agent's utility function $\ua^{M, r}$ under $M$ and $r$ inductively such that
\begin{align*}
    \ua^{M, r}(h, s) & = \sum_a \pi(r(h), r(h, s), a) \cdot \left(\va_{|h| + 1}(s, a) + \sum_{s'} \trans_{|h| + 1}(s, a, s') \cdot \ua^{M, r}(h + (s, a), s')\right) \\
    & \quad - p(r(h), r(h, s)),
\end{align*}
where $\ua^{M, r}(h, s) = 0$ for all $h \in \cH_T$ and $s \in \cS$.
And let $\ua^{M, r}(\emptyset)$ be the overall utility of the agent, i.e.,
\[
    \ua^{M, r}(\emptyset) = \sum_s \trans_0(s) \cdot \ua^{M, r}(\emptyset, s).
\]
In words, $\ua^{M, r}$ is the utility function of the agent applying the reporting strategy $r$ in response to the mechanism $M$.
The mechanism $M$ is IC iff for any such reporting strategy $r$,
\[
    \ua^M(\emptyset) \ge \ua^{M, r}(\emptyset).
\]
Since the revelation principle holds in dynamic environments (see, e.g., \cite{myerson1986multistage}), we focus on IC mechanisms in the rest of the paper.\footnote{Of course, the revelation principle will not hold in our dynamic setting if we allow it to generalize a static setting in which the revelation principle does not hold.  For example, in the case of partial verification --- not every type being able to misreport every other type --- or costly misreporting, the revelation principle is known to hold only under certain conditions~\cite{kephart2016revelation}.  In this paper, we only consider the standard mechanism design setting in which every type can freely misreport any other type, but our techniques can be generalized to the other settings as well.}

\paragraph{Individually-rational mechanisms.}
When payments are allowed, it is standard to impose individual-rationality (IR) (also known as voluntary-participation) constraints on the mechanism, which roughly say that the agent should never be made worse off by participating in the mechanism.
In this paper, we consider two versions of IR constraints:
\begin{itemize}
    \item A mechanism $M$ is overall IR if the overall utility of the agent is nonnegative, i.e., $\ua^M(\emptyset) \ge 0$.
    This ensures that the agent is willing to participate in the overall mechanism.
    \item A mechanism $M$ is dynamic IR if the onward utility of the agent is nonnegative for every history $h$ and current state $s$, i.e., $\ua^M(h, s) \ge 0$ for all $h \in \cH$ and $s \in \cS$.
    This stronger notion of IR further ensures that the agent has no incentive to leave the mechanism at any time.
\end{itemize}
As discussed in later sections, our algorithms work for all $3$ cases regarding IR constraints: no IR (which results in an unbounded objective value if payments are allowed and valued by the principal), overall IR, and dynamic IR.

\section{Computation of Optimal Mechanism: the General Case}
\label{sec:general}

In this section, we investigate the computational problem of finding an optimal dynamic mechanism, which maximizes the principal's overall utility.
For concreteness, we assume that all components of the dynamic environment, including the time horizon $T$, state and action spaces $\cS$ and $\cA$, valuation functions $\vp$ and $\va$, and transition operator $\trans$, are given explicitly as input.

\subsection{Hardness Result for Long-Horizon Environments}
\label{sec:long-horizon}

First we show that the problem with an arbitrarily large time horizon $T$ is intractable.
In general, it takes exponentially many parameters in $T$ to describe a dynamic mechanism, which immediately rules out the possibility of computing a flat representation of an optimal mechanism in polynomial time.
However, this leaves the possibility of computing succinct representations, e.g., an oracle which maps extended histories to distributions over actions.
Our hardness result shows that it is hard to approximate the principal's maximum utility within a constant factor, which rules out the possibility of such succinct representations that can be efficiently evaluated, assuming $\mathsf{P} \ne \mathsf{NP}$.
The proof of the theorem, as well as all other proofs, are deferred to the appendices.

\begin{theorem}
\label{thm:long}
    When the time horizon $T$ can be arbitrarily large, it is $\mathsf{NP}$-hard to approximate the principal's maximum utility within a factor of $7/8 + \eps$ for any $\eps > 0$.
\end{theorem}

\subsection{Algorithm for Short-Horizon Environments}
\label{sec:short-horizon}

Now we give a polynomial-time algorithm for computing an optimal mechanism when $T$ is a constant.
Our algorithm is based on a delicate linear program (LP) formulation, which relies on the following notation and concepts.

\paragraph{Feasible history-state pairs.}
A history-state pair $(h, s)$, where $h = (s_1, a_1, \dots, s_t, a_t)$, is $i$-feasible if $\trans_j(s_j, a_j, s_{j + 1}) > 0$ for every $j \in \{i, i + 1, \dots, t - 1\}$, and $\trans_t(s_t, a_t, s) > 0$.
In other words, starting from $s_i$ and taking the actions specified in $h$, there is a positive probability that the rest of the history and the state $s$ are generated from the transition operator.
We say a pair $(h, s)$ is feasible if it is $1$-feasible.

\paragraph{Feasible extensions.}
For two history-state pairs $(h, s)$ and $(h', s')$ where $h= (s_1, a_1, \dots, s_t, a_t)$ and $h' = (s'_1, a'_1, \dots, s'_{t'}, a'_{t'})$, we say that $(h', s')$ feasibly extends $(h, s)$, i.e., $(h, s) \subseteq (h', s')$, if $(h, s) = (h', s')$, or the following conditions hold simultaneously:
\begin{itemize}
    \item $t = |h| < |h'| = t'$.
    \item For any $i \in [t]$, $(s_i, a_i) = (s'_i, a'_i)$ (this holds automatically when $h = \emptyset$ and therefore $|h| = 0$).
    \item $s = s'_{t + 1}$.
    \item $(h', s')$ is $(|h| + 1)$-feasible (note that this does not require $h$ itself to be feasible).
\end{itemize}

\paragraph{Extended transition operator.}
For notational simplicity we define the following extended transition operator $\etrans_t: \cS \times \cA \to \Delta(\cS)$ for all $t \in \{0\} \cup [T]$, such that
\[
    \etrans_t(s, a, s') = \begin{cases}
        \trans_t(s, a, s'), & \text{if } \trans_t(s, a, s') > 0 \\
        1, & \text{otherwise}
    \end{cases}.
\]
In words, the extended transition operator assigns phantom probability $1$ to each way of transitioning that happens with probability $0$ (so $\etrans_t(s, a)$ does not always normalize to $1$).
As a shorthand, let $\etrans_0(s') = \etrans_0(s, a, s')$ for some $s \in \cS$ and $a \in \cA$ (the specific choice does not matter).
The extended transition operator helps in constructing the flow and IC constraints below and simplifies the formulation.
In particular, we always have $\etrans_t(s, a, s') > 0$.

\begin{figure}[t!]
    \begin{framed}
    \centering
    \begin{flalign}
        \text{\bf objective:} \quad & \max \sum_{h \in \cH, s \in \cS: (h, s)\text{ is feasible}} \left(\sum_{a \in \cA} \vp_{|h| + 1}(s, a) \cdot x(h, s, a) + y(h, s)\right) && \label{eq:obj}
    \end{flalign}
    \begin{flalign}
        \text{\bf flow constraints:} \quad  & z(h, s) = \sum_{a \in \cA} x(h, s, a) && \forall h \in \cH, s \in \cS \label{eq:flow1} \\
        & z(\emptyset, s) = \etrans_0(s) && \forall s \in \cS \label{eq:flow2} \\
        & z(h + (s, a), s') = \etrans_{|h| + 1}(s, a, s') \cdot x(h, s, a) && \forall h \in \cH, s, s' \in \cS, a \in \cA \label{eq:flow3}
    \end{flalign}
    \begin{flalign}
        \text{\bf utility:} \quad & u(h, s) = \sum_{h' \in \cH, s' \in \cS: (h, s) \subseteq (h', s')} \left(\sum_{a \in \cA} \va_{|h'| + 1}(s', a) \cdot x(h', s', a) - y(h', s') \right) && \forall h \in \cH, s \in \cS \label{eq:util}
    \end{flalign}
    \begin{flalign}
        \text{\bf IC constraints:} \quad & u(h, s, s') = \sum_{a \in \cA} \va_{|h| + 1}(s, a) \cdot x(h, s', a) - y(h, s') \nonumber \\
        & \phantom{u(h, s, s') =\ } + \sum_{a \in \cA, s'' \in \cS} \frac{\trans_{|h| + 1}(s, a, s'')}{\etrans_{|h| + 1}(s', a, s'')} \cdot u(h + (s', a), s'') && \forall h \in \cH, s, s' \in \cS \label{eq:ic1} \\
        & u(h, s) \ge \frac{\etrans_{|h|}(s_p, a_p, s)}{\etrans_{|h|}(s_p, a_p, s')} \cdot u(h, s, s'), \text{where}\, (s_p, a_p) = \last(h) && \forall h \in \cH, s, s' \in \cS \label{eq:ic2}
    \end{flalign}
    \begin{flalign}
        \text{\bf IR constraints:} \quad & u(h, s) \ge 0 && \forall h \in \cH, s \in \cS \label{eq:ir}
    \end{flalign}
    \begin{flalign}
        \text{\bf feasible actions:} \quad & x(h, s, a) \ge 0 && \forall h \in \cH, s \in \cS, a \in \cA \label{eq:actions}
    \end{flalign}
    \begin{flalign}
        \text{\bf feasible payments:} \quad & y(h, s) \ge 0 && \forall h \in \cH, s \in \cS \label{eq:payments}
    \end{flalign}
    \caption{Linear program for computing an optimal dynamic mechanism.}
    \label{fig:lp}
    \end{framed}
\end{figure}


\paragraph{Last state-action pair.}
For a history $h \in \cH$ where $h = (s_1, a_1, \dots, s_t, a_t)$, we use $\last(h)$ as a shorthand for the last state-action pair, i.e.,
    $\last(h) = (s_t, a_t)$.
In particular, when $h = \emptyset$, $\last(h)$ can be any state-action pair (the choice does not affect our results --- it is merely a simplifying shorthand).

Now we are ready to describe the LP formulation.
The complete formulation is given in Figure~\ref{fig:lp}.
The formulation is for nonnegative payments and dynamic IR constraints --- we will discuss later how the formulation can be modified to allow other types of constraints.
Below, we describe each of its components.

\paragraph{Variables, flow constraints, and the corresponding mechanism.}
There are $5$ classes of variables in the LP:
\begin{itemize}
    \item $x(h, s, a)$: the absolute, unconditional probability that the mechanism reaches state $s$ via history $h$, and takes action $a$.
    \item $y(h, s)$: the payment for history-state pair $(h, s)$, scaled by the probability that the mechanism reaches $s$ via $h$ (i.e., $z(h, s)$).
    \item $z(h, s)$: the probability that the mechanism reaches state $s$ via history $h$, which by definition satisfies
    \[
        z(h, s) = \sum_{a \in \cA} x(h, s, a).
    \]
    \item $u(h, s)$: the onward utility of the agent at state $s$ with history $h$ assuming truthful reporting, scaled by the probability that the mechanism reaches $s$ via $h$ (i.e., $z(h, s)$).
    \item $u(h, s, s')$: the onward utility of the agent at state $s$ with history $h$ if the agent misreports $s'$, assuming truthful reporting in the future, scaled by the probability that the mechanism reaches $s'$ via $h$ (i.e., $z(h, s')$).
\end{itemize}
The flow constraints (Eq.~\eqref{eq:flow1}-\eqref{eq:flow3}) enforce roughly the above interpretation of variables to $x(h, s, a)$ and $z(h, s)$, except for ways of transition that have probability $0$.
For each way of transition with probability $0$, the extended transition operator assigns phantom probability $1$.
This phantom probability is not counted in the objective function (because only feasible history-state pairs are counted) or in the utility variables $u(h, s)$ (because only feasible extensions are counted).
So, the phantom probability does not affect the principal's or the agent's utility assuming truthful reporting.
Instead, together with other constraints, it guarantees that the mechanism behaves well even for history-state pairs that appear with probability $0$ under truthful reporting, which is necessary for the mechanism to be IC (see later paragraphs).
Under the above interpretation, the LP variables (and in particular, $x(h, s, a)$, $y(h, s)$ and $z(h, s)$) naturally correspond to a mechanism $M = (\pi, p)$.
Formally, for each $h \in \cH$, $s \in \cS$:
\begin{itemize}
    \item If $z(h, s) > 0$, then
    \[
        p(h, s) = y(h, s) / z(h, s),
    \]
    and for each $a \in \cA$,
    \[
        \pi(h, s, a) = x(h, s, a) / z(h, s).
    \]
    \item If $z(h, s) = 0$, then let $\pi(h, s)$ be an arbitrary distribution over $\cA$, and $p(h, s) = 0$.
\end{itemize}
The feasibility of the mechanism (i.e., every $\pi(h, s)$ is a distribution over $\cA$ and every $p(h, s)$ is nonnegative) is guaranteed by constraints \eqref{eq:flow1}, \eqref{eq:actions} and \eqref{eq:payments}.
We remark that while the mechanism constructed from the LP variables may not be unique, effectively this makes no difference, since the parts of the mechanism that are chosen arbitrarily can never be accessed when executing the mechanism.
This is because $z(h, s) = 0$ only if at some point in the history $h$, there is an action that the mechanism would never play given the reported states and actions before that.
In particular, the above does not simply apply to all history-state pairs $(h, s)$ that are reached with probability $0$ under truthful reporting, in which case $z(h, s)$ may still be positive due to the extended transition operator.
Moreover, given any mechanism, one can construct LP variables in a similar way, such that the mechanism constructed from these variables is the same as the original mechanism (modulo the unreachable parts).
In other words, the above correspondence is effectively bijective.

\paragraph{The objective.}
The objective function of the LP (Eq.~\eqref{eq:obj}) is precisely the overall utility of the principal under the mechanism constructed above, assuming truthful reporting.
This is captured by the following lemma.

\begin{lemma}
\label{lem:obj}
    Let $M = (\pi, p)$ be the mechanism constructed from variables $x(h, s, a)$, $y(h, s)$, and $z(h, s)$ which satisfy the flow constraints.
    Then
    \[
        \up^M(\emptyset) = \sum_{h \in \cH, s \in \cS: (h, s)\text{ is feasible}} \left(\sum_{a \in \cA} \vp_{|h| + 1}(s, a) \cdot x(h, s, a) + y(h, s)\right).
    \]
\end{lemma}

From this lemma, it is clear that the objective of the LP is the natural quantity to maximize.

\paragraph{Utility.}
The utility constraints (Eq.~\eqref{eq:util}) collect the agent's onward utility, where $u(h, s)$ is equal to the agent's onward utility in state $s$ from history $h$, assuming truthful reporting, scaled by $z(h, s)$.
This is captured by the following lemma.

\begin{lemma}
\label{lem:util}
    Let $M = (\pi, p)$ be the mechanism constructed from variables $x(h, s, a)$, $y(h, s)$, and $z(h, s)$ which satisfy the flow and utility constraints.
    For all $h \in \cH$, $s \in \cS$, 
    \[
        u(h, s) = z(h, s) \cdot \ua^M(h, s).
    \]
\end{lemma}

The proof of Lemma~\ref{lem:util} is essentially the same as that of Lemma~\ref{lem:obj}.
Given the correspondence to the agent's utility $\ua^M(h, s)$, the utility variables $u(h, s)$ act as auxiliary variables in IC constraints.

\paragraph{IC constraints.}
IC constraints are a key component of the LP formulation.
There are two families of IC constraints: collecting the agent's scaled utility from single-step misreporting (Eq.~\eqref{eq:ic1}), and subsequently restricting the mechanism such that there is no incentive for misreporting (Eq.~\eqref{eq:ic2}).
In Eq.~\eqref{eq:ic1}, we build variables $u(h, s, s')$, which is supposed to be the onward utility of the agent in state $s$ from history $h$ misreporting $s'$, assuming truthful reporting in the future, scaled by $z(h, s')$ (rather than $z(h, s)$).
This is captured by the following lemma.

\begin{lemma}
\label{lem:ic1}
    Let $M = (\pi, p)$ be the mechanism constructed from variables $x(h, s, a)$, $y(h, s)$, and $z(h, s)$ which satisfy the flow constraints, the utility constraints, and Eq.~\eqref{eq:ic1}.
    Then the following statement holds: for all $h \in \cH$, $s, s' \in \cS$, let reporting strategy $r_{h, s, s'}$ be such that
    \[
        r_{h, s, s'}(h', s'') = \begin{cases}
            s', & \text{if } h = h' \text{ and } s = s'' \\
            s'', & \text{otherwise}
        \end{cases}.
    \]
    That is, $r_{h, s, s'}$ misreports $s'$ only in state $s$ from history $h$, and reports truthfully otherwise.
    Then for all $h \in \cH$, $s, s' \in \cS$,
    \[
        u(h, s, s') = z(h, s') \cdot \ua^{M, r_{h, s, s'}}(h, s).
    \]
\end{lemma}

Given Lemma~\ref{lem:ic1}, Eq.~\eqref{eq:ic2} then guarantees that the mechanism $M$ is robust against single-step misreporting for all reachable history-state pairs.
\begin{lemma}
\label{lem:ic2}
    Let $M = (\pi, p)$ be the mechanism constructed from variables $x(h, s, a)$, $y(h, s)$, and $z(h, s)$ which satisfy the flow constraints, the utility constraints, and Eq.~\eqref{eq:ic1}.
    The following is true if and only if the LP variables also satisfy Eq.~\eqref{eq:ic2}: for all $h \in \cH$, $s, s' \in \cS$ where $(h, s)$ is reachable by the mechanism $M$,
    \[
        \ua^M(h, s) \ge \ua^{M, r_{h, s, s'}}(h, s).
    \]
\end{lemma}

We then show that a mechanism is IC if and only if there is no incentive for single-step misreporting, which directly implies that the mechanism $M$ constructed from the LP variables is IC.
This is captured by the following lemma.
\begin{lemma}
\label{lem:ic}
    Let $M = (\pi, p)$ be the mechanism constructed from variables $x(h, s, a)$, $y(h, s)$, and $z(h, s)$ which satisfy the flow constraints, the utility constraints, and Eq.~\eqref{eq:ic1}.
    Then $M$ is IC if and only if the LP variables also satisfy Eq.~\eqref{eq:ic2}.
\end{lemma}

\paragraph{IR constraints, feasible actions, and feasible payments.}
These constraints are straightforward given the correspondence between the LP variables and the mechanism that we have discussed above.
Note that while Eq.~\eqref{eq:ir} is for dynamic IR (i.e., the agent has no incentive to leave the mechanism at any point) and Eq.~\eqref{eq:payments} is for nonnegative payments, it is easy to replace them with similar constraints that correspond to overall IR or no payments.
See Section~\ref{sec:customizing} for more details.

\paragraph{Optimality of LP solution.}
Given the above facts, we are ready to state and prove the main result of the paper.

\begin{theorem}
\label{thm:short}
    There is an algorithm which computes an optimal IC and (optionally) IR dynamic mechanism, with or without payments, in time $O(\mathrm{poly}(|\cS|^T, |\cA|^T, L))$, where $L$ is the number of bits required to encode each of the input parameters.
    In particular, when $T$ is constant, the algorithm runs in polynomial time.
\end{theorem}

\subsection{Customizing the LP Formulation.}
\label{sec:customizing}

The LP formulation in Figure~\ref{fig:lp} allows for nonnegative payments, assumes that the principal cares about payments as much as the agent, and enforces dynamic IR constraints.
As mentioned above, one can customize all these components by modifying the corresponding parts of the LP formulation.
Below we discuss several ways of customization.
\begin{itemize}
    \item {\bf Unequal valuations for payments}: in the case where the principal has utility $c$ for one unit of payment (whereas without loss of generality the agent has utility $1$), one may replace the objective function (Eq.~\eqref{eq:obj}) with
    \[
        \max \sum_{h \in \cH, s \in \cS: (h, s)\text{ is feasible}} \left(\sum_{a \in \cA} \vp_{|h| + 1}(s, a) \cdot x(h, s, a) + c \cdot y(h, s)\right).
    \]
    Note that our formulation works only when the principal cares linearly about payments.
    Notably, the principal may not care about payments at all (as in the case of paying the agent in ``brownie points''), or even dislike payments made by the agent (as in the case where the agent is asked to expend useless effort or ``burn money'' and the principal cares in part about the resulting loss of welfare).
    \item {\bf No payments}: to forbid payments in the mechanism, one can simply replace Eq.~\eqref{eq:payments} with
    \[
        y(h, s) = 0,\ \forall h \in \cH, s \in \cS.
    \]
    \item {\bf Feasible intervals of payments}: more generally, one may wish to specify a feasible interval $[a_{h, s}, b_{h, s}]$ for the payment at each history-state pair $(h, s)$ such that $a_{h, s} \le p(h, s) \le b_{h, s}$, which subsumes both nonnegative payments and no payments as special cases.
    This can be done by replacing Eq.~\eqref{eq:payments} with
    \[
        a_{h, s} \cdot z(h, s) \le y(h, s) \le b_{h, s} \cdot z(h, s),\ \forall h \in \cH, s \in \cS.
    \]
    \item {\bf Overall/no IR}: when the agent can choose whether to participate in the mechanism, but cannot leave halfway (corresponding to an overall IR constraint), one can replace Eq.~\eqref{eq:ir} with
    \[
        \sum_{s \in \cS} u(\emptyset, s) \ge 0.
    \]
    Also, when leaving the mechanism is not an option for the agent from the very beginning (corresponding to no IR constraint), one may remove IR constraints simply by removing Eq.~\eqref{eq:ir}.
    \item {\bf Discount factors}: to accommodate the case where the agent has a discount factor $0 \le \delta < 1$, one can modify the LP formulation in the following way:
    \begin{itemize}
        \item Replace Eq.~\eqref{eq:util} with
        \[
            u(h, s) = \sum_{h' \in \cH, s' \in \cS: (h, s) \subseteq (h', s')} \delta^{|h'| - |h|} \cdot \left(\sum_{a \in \cA} \va_{|h'| + 1}(s', a) \cdot x(h', s', a) - y(h', s') \right),\ \forall h \in \cH, s \in \cS.
        \]
        \item Replace Eq.~\eqref{eq:ic1} with
        \begin{align*}
            u(h, s, s') & = \sum_{a \in \cA} \va_{|h| + 1}(s, a) \cdot x(h, s', a) - y(h, s') \\
            \phantom{u(h, s, s')} & \phantom{\ = \ } + \delta \cdot \sum_{a \in \cA, s'' \in \cS} \frac{\trans_{|h| + 1}(s, a, s'')}{\etrans_{|h| + 1}(s', a, s'')} \cdot u(h + (s', a), s''),\ \forall h \in \cH, s, s' \in \cS \\
        \end{align*}
    \end{itemize}
    \item {\bf Deterministic mechanisms}: the problem of computing an optimal deterministic mechanism is $\mathsf{NP}$-hard even in static environments \cite{conitzer2002complexity,conitzer2004self}.
    Nevertheless, given our LP formulation, one can restrict the mechanism to be deterministic by introducing Boolean variables, resulting in a mixed integer LP.  While integer LPs are hard to solve in a worst-case sense, real-world problems often admit certain structures which can be exploited by commercial solvers such as CPLEX and Gurobi.
    To be specific, we introduce a Boolean variable $c(h, s, a)$ which controls $x(h, s, a)$ for all $h \in \cH$, $s \in \cS$, and $a \in \cA$, and ensures that fixing $h$ and $s$, $x(h, s, a)$ can be positive for at most one action $a \in \cA$.
    This is implemented by the following constraints (in addition to the existing ones):
    \begin{align*}
        & x(h, s, a) \le c(h, s, a) && \forall h \in \cH, s \in \cS, a \in \cA \\
        & \sum_{a \in \cA} c(h, s, a) = 1 && \forall h \in \cH, s \in \cS \\
        & c(h, s, a) \in \{0, 1\} && \forall h \in \cH, s \in \cS, a \in \cA.
    \end{align*}
\end{itemize}
We also remark that the above discussion is non-exhaustive: one can impose richer restrictions by modifying the LP formulation in other linear ways, and/or combining the above modifications.

\section{The Case of Myopic Agents: Characterization and Faster Algorithm}
\label{sec:myopic}

In this section, we consider a special case of the problem of computing optimal dynamic mechanisms, namely the case where the agent is myopic, or, equivalently, the agent has a discount factor of $0$.
While our LP-based algorithm still applies, as we will see below, optimal mechanisms for myopic agents enjoy a succinct representation in this case, which also enables a faster algorithm that scales only linearly in the time horizon $T$.

\paragraph{Myopic agents.}
The utility $\ua^M$ of a myopic agent under mechanism $M$ is such that
\[
    \ua^M(h, s) = \sum_a \pi(h, s, a) \cdot \va_{|h| + 1}(s, a) - p(h, s),
\]
where $\ua^M(h, s) = 0$ for all $h \in \cH_T$ and $s \in \cS$.
Given a reporting strategy $r$, the utility $\ua^{M, r}$ of the agent under mechanism $M$ and reporting strategy $r$ is
\[
    \ua^{M, r}(h, s) = \sum_a \pi(r(h), r(h, s), a) \cdot \va_{|h| + 1}(s, a) - p(r(h), r(h, s)).
\]
$M$ is IC if and only if for all $h \in \cH$ and $s \in \cS$, there are no future reporting strategies that lead to better utility, i.e., for every reporting strategy $r$ where $r(h', s') = s'$ whenever $|h'| < |h|$,
\[
    \ua^M(h, s) \ge \ua^{M, r}(h, s).
\]
Note that since the agent is myopic, it is insufficient to simply require $\ua^M(\emptyset) \ge \ua^{M, r}(\emptyset)$.
Also, it is necessary to restrict misreporting to the future, since otherwise the agent would be allowed and sometimes incentivized to change the past, leading to unrealistically strong IC requirements.
Again, since the revelation principle holds, we focus only on IC mechanisms.

\subsection{Characterization of Optimal Mechanisms}

We first show that when the agent is myopic, without loss of generality, the actions and payments specified by an optimal mechanism depend only on the time, the previous state, the previous action and the current state (we call such a mechanism a {\em succinct mechanism}), instead of the entire history-state pair.

\begin{lemma}
\label{lem:characterization}
    Fix a dynamic environment.
    When the agent is myopic, for any IC mechanism $M = (\pi, p)$, there is another IC mechanism $M' = (\pi', p')$ (which is IR whenever $M$ is) such that
    \begin{itemize}
        \item $\up^{M'}(\emptyset) \ge \up^M(\emptyset)$, and
        \item for all $h \in \cH$, $s \in \cS$, $\pi'$ and $p'$ depend only on $|h|$, $s_p$, $a_p$ and $s$, where $(s_p, a_p) = \last(h)$.
    \end{itemize}
    Moreover, the above is true regardless of whether payments are allowed, or which IR constraints are required.
\end{lemma}

\subsection{Faster Algorithm for Myopic Agents}

Based on the above characterization, we present below a faster algorithm for computing an optimal mechanism in the face of a myopic agent.
In particular, the time complexity of this algorithm depends only linearly on the time horizon $T$, making it feasible for dynamic environments with a long time horizon.
This is in contrast with the case of patient agents, for which, as we have seen, the long-horizon problem is hard to approximate.

To improve readability, we use the following shorthand notation for succinct mechanisms.
For a succinct mechanism $M = (\pi, p)$, for any $h \in \cH$ and $s \in \cS$, let $\pi(t,  s_p, a_p, s) = \pi(h, s)$ be the action policy at $(h, s)$, and $p(t, s_p, a_p, s) = p(h, s)$ be the payment function, where $(s_p, a_p) = \last(h)$ and $t = |h| + 1$.
Also, observe that the principal's onward utility at any history-state pair $(h, s)$ depends only on the previous state $s_p$, the previous action $a_p$, and the current state $s$.
In such cases, we also denote this utility by $\up^M(t, s_p, a_p, s) = \up^M(h, s)$.

\begin{algorithm}[t]
\SetAlgoNoLine
\KwIn{Time horizon $T$, transition probabilities $\{\trans_t\}_{t \in [T]}$, principal's valuation functions $\{\vp_t\}_{t \in [T]}$, agent's valuation functions $\{\va_t\}_{t \in [T]}$.}
\KwOut{An optimal IC (for a myopic agent) mechanism $M = (\pi, p)$.}
    \For{$t = T, T - 1, \dots, 1$}{
        \For{$s \in \cS$, $a \in \cA$}{
            let $u(t, s, a) \gets \vp_t(s, a) + \sum_{s' \in \cS} \trans_t(s, a, s') \cdot \up^M(t + 1, s, a, s')$\;
            \tcc{the above operation is well-defined, in particular because $\up^M(t + 1, s, a, s')$ depends only on the part of $M$ that has already been computed}
        }
        \For{$s_p \in \cS, a_p \in \cA$}{
            let $(\pi', p') \gets \optstatmech(\cS, \cA, \{\trans_{t - 1}(s_p, a_p, s)\}_s, \{u(t, s, a)\}_{s, a}, \{\va_t(s, a)\}_{s, a})$\;
            \tcc{call $\optstatmech$ to compute an optimal static mechanism $(\pi', p')$, in a static environment with type space $\cS$, action space $\cA$, population distribution $\{\trans_{t - 1}(s_p, a_p, s)\}_s$, principal's utility function $\{u(t, s, a)\}_{s, a}$, and agent's utility function $\{\va_t(s, a)\}_{s, a}$}
            \For{$s \in \cS$}{
                let $\pi(t, s_p, a_p, s) \gets \pi'(s)$, and $p(t, s_p, a_p, s) \gets p'(s)$\;
            }
        }
    }
    return $M = (\pi, p)$\;
\caption{Algorithm for computing an optimal mechanism against a myopic agent.}
\label{alg:myopic}
\end{algorithm}

The full algorithm is given as Algorithm~\ref{alg:myopic}.
It uses as a subroutine an algorithm $\optstatmech$ which computes an optimal IC (and optionally IR) mechanism in static environments, with or without payments.
It is known that such an algorithm can be implemented using linear programming, and in some cases in more efficient ways \cite{conitzer2002complexity,conitzer2006computing,zhang2021automated}.
Algorithm~\ref{alg:myopic} proceeds in an inductive fashion, building a succinct mechanism backwards, one layer at a time.
It repeatedly solves the problem of maximizing the principal's expected onward utility over the current state $s$, given the previous state $s_p$ and the previous action $a_p$.
Since $s_p$ and $a_p$ together induce a roll-in distribution over the state space, this problem can be reduced to computing an optimal static mechanism, where the valuation function of the principal depends on the optimal mechanism in the following layers.
This can then be solved by calling $\optstatmech$, the algorithm for computing an optimal static mechanism.
Below we state and prove the correctness and computational efficiency of Algorithm~\ref{alg:myopic}.

\begin{theorem}
\label{thm:myopic}
    When the agent is myopic, Algorithm~\ref{alg:myopic} computes an optimal IC and (optionally) IR dynamic mechanism, with or without payments, in time
    \[
        O(T |\cS| |\cA| \cdot T_\mathrm{stat}(|\cS|, |\cA|, L)) = O(T \cdot \mathrm{poly}(|\cS|, |\cA|, L)),
    \]
    where $T_\mathrm{stat}$ is the time complexity of $\optstatmech$, and $L$ is the number of bits required to encode each of the input parameters.
\end{theorem}

\paragraph{Customizing Algorithm~\ref{alg:myopic}.}
We remark that Algorithm~\ref{alg:myopic} can also be customized to allow for unequal valuations of payments, feasible intervals of payments, etc.
Moreover, it can be adapted to compute an optimal deterministic mechanism, by requiring $\optstatmech$ to compute an optimal deterministic static mechanism.
Again, while this is generally hard to compute, for practical purposes, it is reasonable to expect that $\optstatmech$ implemented using commercial mixed integer LP solvers (or in other practically efficient ways) can find an optimal mechanism efficiently.

\section{Infeasibility of Memoryless Mechanisms}
\label{sec:memoryless}

From a planning perspective, automated dynamic mechanism design can be viewed equivalently as planning in MDPs where the current state cannot be directly observed, but instead, has to be reported by a strategic agent whose interest may not align with the planner's.
In particular, when the planner and the agent share the same valuation function, automated dynamic mechanism design degenerates to the classical problem of planning in episodic MDPs with a finite planning horizon.
In the latter problem, it is well known that without loss of generality, any optimal policy depends only on the time and the current state, i.e., it is memoryless.
And moreover, such optimal policies can be computed in polynomial time.
In light of the above facts, the following questions arise naturally: are there (approximately) optimal mechanisms that are also memoryless, and can we find optimal memoryless mechanisms efficiently?
In this section, we give negative answers to both questions, which means memoryless mechanisms are generally infeasible for dynamic environments.
We first show that memoryless mechanisms can be arbitrarily worse than general, history-dependent mechanisms, against both patient and myopic agents.

\begin{theorem}
\label{thm:suboptimality}
    Regardless of whether the agent is myopic, for any $\eps > 0$, there is a dynamic environment where the principal's utility under an optimal memoryless mechanism is at most an $\eps$ fraction of the principal's optimal utility.
\end{theorem}

Now we show that on top of the suboptimality, optimal memoryless mechanisms are computationally hard to approximate.
\begin{theorem}
\label{thm:np-hardness}
    Regardless of whether the agent is myopic, it is $\mathsf{NP}$-hard to approximate the principal's maximum utility under memoryless mechanisms within a factor of $7/8 + \eps$ for any $\eps > 0$.
\end{theorem}

\section{Experimental Results}
\label{sec:exp}

In this section, we present experimental results where our algorithms are applied to synthetic dynamic environments of different characteristics.
The main goals of the experiments are
\begin{itemize}
    \item to provide a proof of concept for the methods proposed in this paper,
    \item to illustrate the necessity of considering incentives when planning in dynamic environments (as opposed to disregarding the agent's valuations and treating the problem simply as an MDP based on the principal's valuations),
    \item to study the effect of cooperation and competition in dynamic mechanism design, and
    \item to understand the difference between patient and myopic agents from the principal's perspective, especially when the parameters of the environment vary.
\end{itemize}

\subsection{Setup of Experiments}

\paragraph{Mechanisms/models of the agent under consideration.}
For each dynamic environment examined, we consider the following quantities from different combinations of mechanisms and models of the agent:
\begin{itemize}
    \item {\bf Na\"ive mechanisms facing a na\"ive agent}: the principal's optimal utility facing a na\"ive agent who always reports truthfully, i.e., the optimal utility when treating the problem simply as an MDP based on the principal's valuations.
    \item {\bf Na\"ive mechanisms facing a patient agent}: the principal's utility, when executing the optimal mechanism/policy for na\"ive agents, facing a strategic agent who is patient.
    \item {\bf Na\"ive mechanisms facing a myopic agent}: the principal's utility, when executing the optimal mechanism/policy for na\"ive agents, facing a strategic agent who is myopic.
    \item {\bf Patient mechanisms facing a patient agent}: the principal's optimal utility facing a strategic agent who is patient.
    \item {\bf Myopic mechanisms facing a myopic agent}: the principal's optimal utility facing a strategic agent who is myopic.
\end{itemize}
For simplicity, payments are not allowed in any of our experiments.

\paragraph{Dynamic environments.}
To manifest the effect of cooperation and competition, we generate synthetic dynamic environments in the following way:
\begin{itemize}
    \item Fix the time horizon $T$, number of states $|\cS|$, number of actions $|\cA|$, and correlation parameter $\eta \in [-1, 1]$ (explained below).
    \item Let the initial distribution $\trans_0$ be a random distribution generated in the following way: for each state $s$, we generate a uniformly random real number $\rand(s)$ between $0$ and $1$, which is proportional to $\trans_0(s)$.
    That is, $\trans_0(s) = \rand(s) / \left(\sum_{s'} \rand(s')\right)$.
    \item For each $t \in [T]$, $s \in \cS$ and $a \in \cA$, we generate the transition distribution $\trans_t(s, a)$ independently in the same way that $\trans_0$ is generated.
    \item For each $t \in [T]$, $s \in \cS$ and $a \in \cA$, let $\vp_t(s, a)$ be an independent, uniformly random real number between $0$ and $1$.
    \item For each $t \in [T]$, $s \in \cS$ and $a \in \cA$, let $\va_t(s, a) = \eta \cdot \vp_t(s, a) + (1 - |\eta|) \cdot \rand(t, s, a)$, where $\rand(t, s, a)$ is an independent, uniformly random real number between $0$ and $1$.
\end{itemize}
The correlation parameter $\eta$ controls the extent to which the interests of the principal and the agent are (mis)aligned.
In particular, if $\eta = 1$, then the principal and the agent have exactly the same valuations, corresponding to full cooperation.
If $\eta = -1$, then the principal and the agent are in a zero-sum situation, corresponding to full competition.

\begin{figure}[t]
    \centering
    \includegraphics[width=\figwidth]{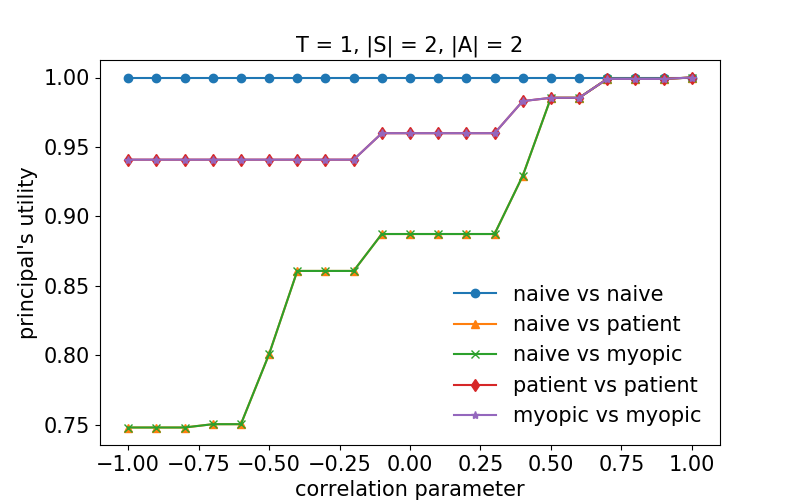}
    \includegraphics[width=\figwidth]{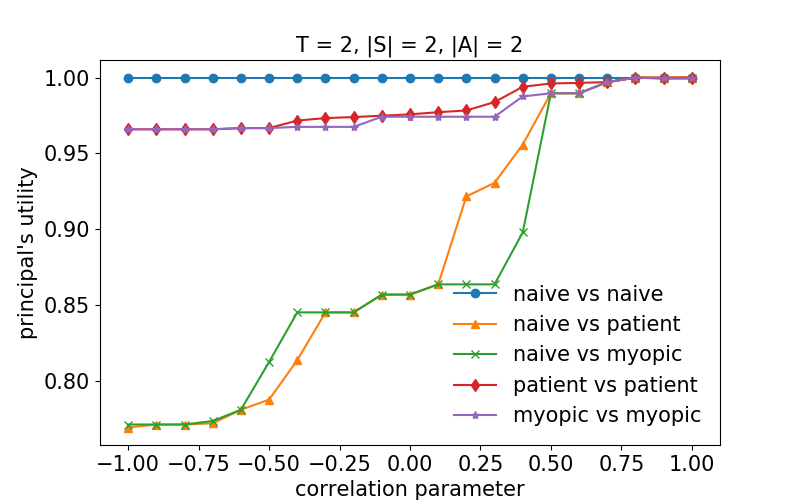}
    \includegraphics[width=\figwidth]{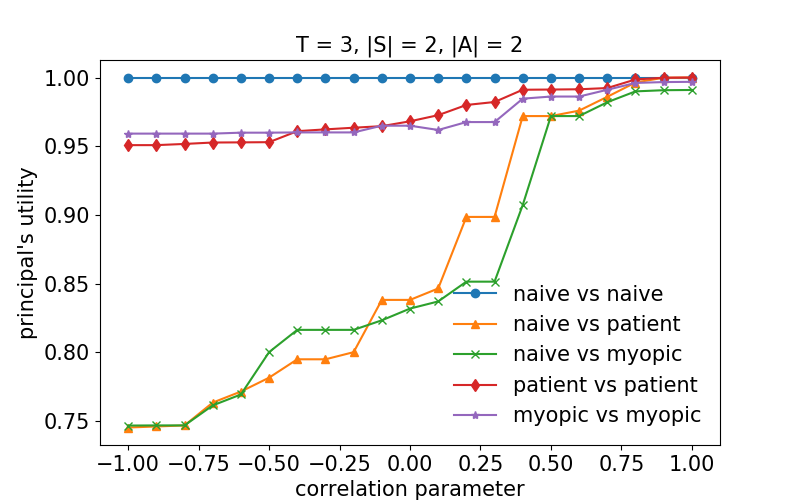}
    \includegraphics[width=\figwidth]{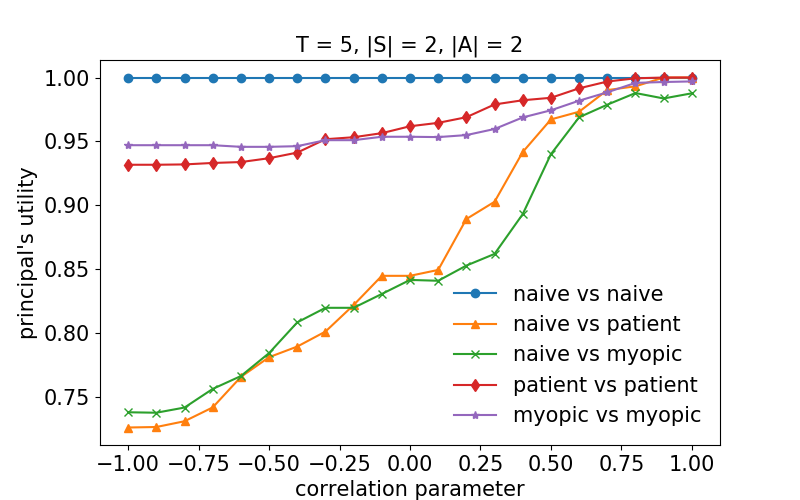}
    \includegraphics[width=\figwidth]{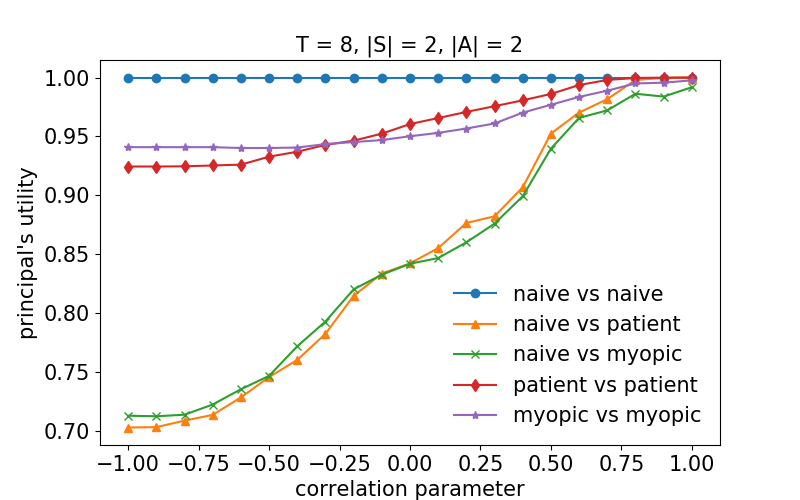}
    \includegraphics[width=\figwidth]{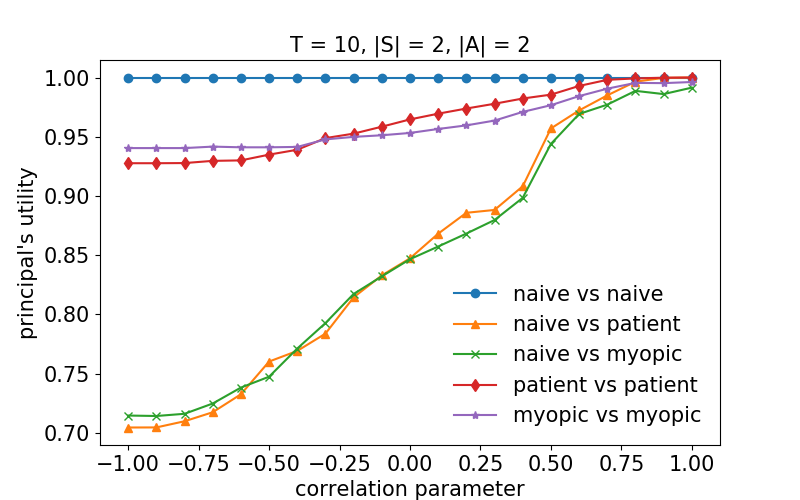}
    \caption{
        Performance of different mechanisms facing different types of agents when $|\cS| = |\cA| = 2$ and the time horizon $T$ varies.
        All numbers are normalized by the optimal utility facing a na\"ive agent.
        Every point is an average of $10$ independent runs using different random seeds.
    }
    \label{fig:exp_time}
\end{figure}

\begin{figure}[t]
    \centering
    \includegraphics[width=\figwidth]{fig/T2_S2_A2}
    \includegraphics[width=\figwidth]{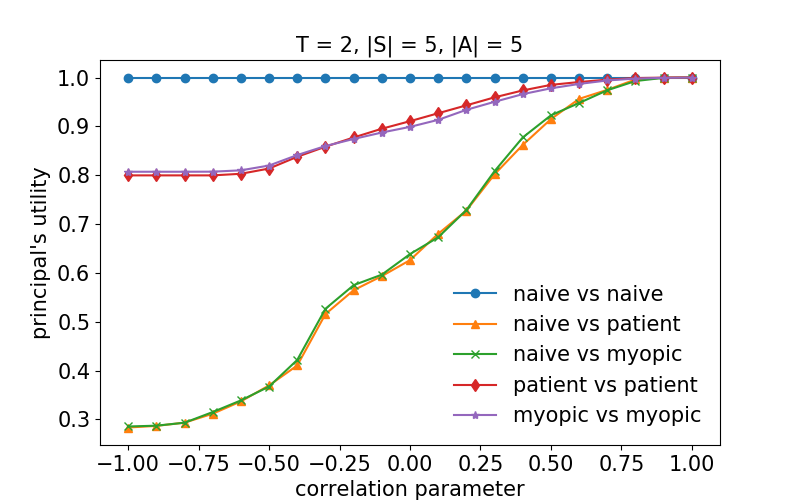}
    \includegraphics[width=\figwidth]{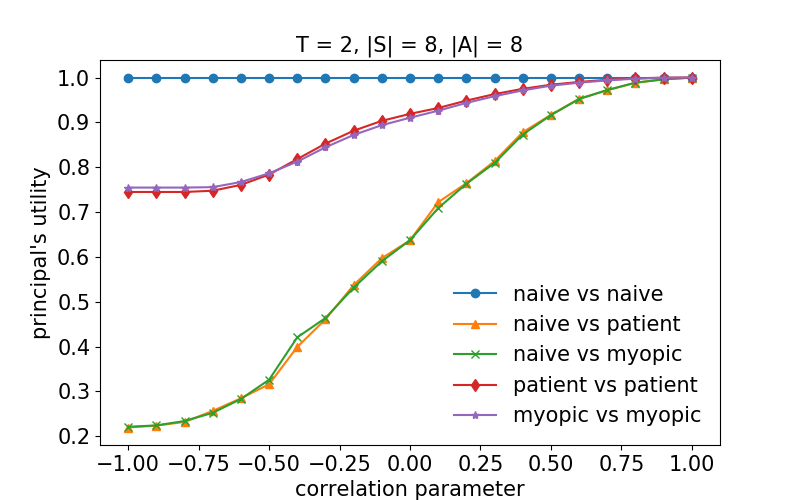}
    \includegraphics[width=\figwidth]{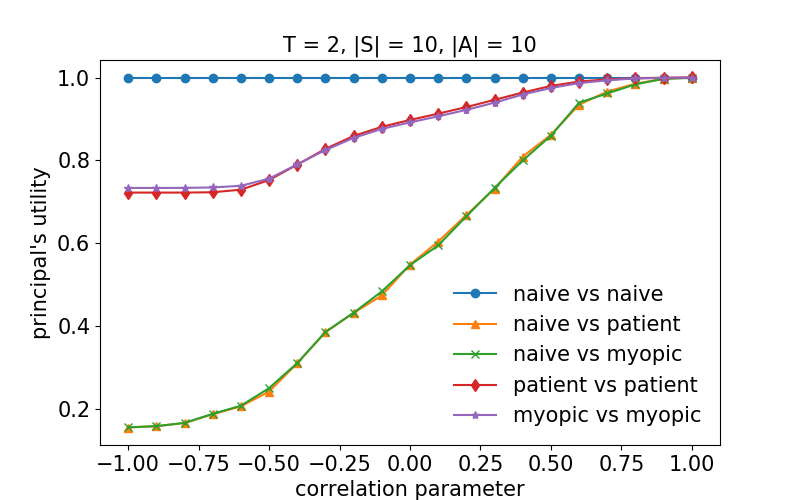}
    \includegraphics[width=\figwidth]{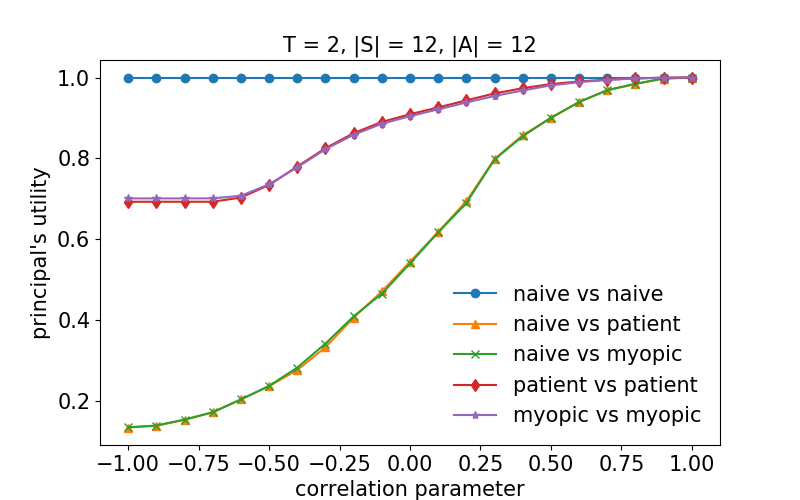}
    \includegraphics[width=\figwidth]{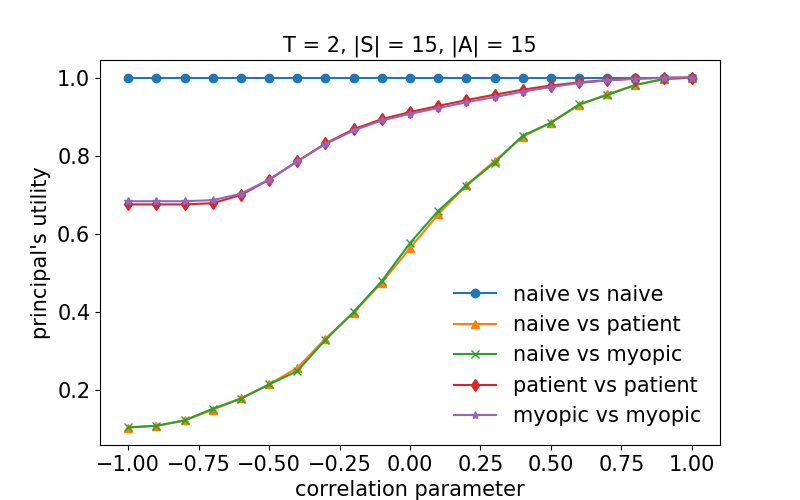}
    \caption{
        Performance of different mechanisms facing different types of agents when $T = 2$ and the numbers of states and actions, $|\cS|$ and $|\cA|$, vary.
        All numbers are normalized by the optimal utility facing a na\"ive agent.
        Every point is an average of $10$ independent runs using different random seeds.
    }
    \label{fig:exp_state}
\end{figure}

\subsection{Summary of Experimental Results}

\paragraph{Suboptimality of na\"ive mechanisms.}
As we can see from Figure~\ref{fig:exp_time}, even when the state and action spaces are extremely simple, i.e., there are only $2$ states and $2$ actions, when the correlation parameter $\eta = -1$ (i.e., when the agent acts adversarially), na\"ive mechanisms facing a strategic agent can only achieve about 75\% of the na\"ive benchmark, i.e., the optimal utility when the agent is na\"ive.
When $\eta = 0$ (i.e., when the agent's and principal's valuations are independent), na\"ive mechanisms facing a strategic agent still achieve only 85\% of the na\"ive benchmark.
On the other hand, the respective optimal mechanisms facing a patient or myopic agent consistently achieve about 95\% of the na\"ive benchmark.
This gap is further amplified in Figure~\ref{fig:exp_state}: as the environment becomes more and more complex (i.e., the numbers of states and actions become larger and larger), the utility of na\"ive mechanisms facing a strategic agent drops below 20\% of the na\"ive benchmark when $\eta = -1$, and to about 50\% when $\eta = 0$.
In contrast, the respective optimal mechanisms facing a patient or myopic agent still achieve about 70\% of the na\"ive benchmark even when $\eta = -1$.
These phenomena suggest that when the agent is not fully cooperative, taking strategic behavior into consideration significantly improves the principal's utility, even in extremely simple dynamic environments.
Moreover, the more complex the environment is, the larger this gap becomes.

Another interesting fact to note is that even when the principal's and the agent's valuations are exactly the same (i.e., when $\eta = 1$), na\"ive mechanisms are still suboptimal facing a myopic agent, since the agent may sacrifice greater long-term gain in exchange for smaller immediate value.
This phenomenon is more significant in Figure~\ref{fig:exp_time}, especially in environments with longer time horizons.
In such cases, taking into consideration the fact that the agent is myopic mitigates the loss, and recovers almost all the utility of the na\"ive benchmark.

\paragraph{Effect of cooperation and competition.}
As the correlation parameter increases, both Figure~\ref{fig:exp_time} and Figure~\ref{fig:exp_state} show clear upward trends in all the quantities that we consider (except for the na\"ive benchmark which is always normalized to $1$), as one would expect.
Nevertheless, we note the following facts from the figures: compared to na\"ive mechanisms, optimal mechanisms facing a strategic agent are much less affected by the correlation parameter.
Moreover, as Figure~\ref{fig:exp_time} shows, the performance of optimal mechanisms facing a strategic agent is remarkably stable as the time horizon grows.
In other words, in random dynamic environments, the utility loss caused by competing interests of the principal and the agent is only mildly amplified by long time horizons.

\paragraph{Difference between patient and myopic agents.}
As can be seen from the figures, regardless of whether the agent is patient or myopic, the principal's optimal utility is almost the same.
Nevertheless, the difference appears to be amplified as the time horizon grows (see Figure~\ref{fig:exp_time}).
When the correlation parameter $\eta = -1$, the optimal utility facing a myopic agent is noticeably larger than that facing a patient agent --- which makes sense as only the patient agent has interests truly opposite those of the principal.
This gap shrinks as $\eta$ becomes larger, and vanishes when $\eta$ is around $-0.25$.
Then, as $\eta$ continues to grow, the optimal utility facing a myopic agent falls behind and never catches up.
In particular, when $\eta = 1$, the optimal utility facing a patient agent is the same as the na\"ive benchmark, whereas that facing a myopic agent is slightly smaller.
The above phenomena indicate that in environments with a long time horizon, myopic agents are easier to exploit, while patient agents are easier to cooperate with.
Interestingly, the critical value of $\eta$, where the optimal utility facing a patient agent catches up, is about $-0.25$ instead of $0$, which suggests that even when the principal's and the agent's valuations are mildly negatively correlated, it is possible to find a middle ground where cooperation is more beneficial than exploitation in the long run.

\section{Conclusion}
We studied automated dynamic mechanism design and showed that, while it is computationally hard to find (even approximately) optimal mechanisms when (1) facing a patient agent and  (2) the horizon is long, when either of these two conditions is dropped, an optimal mechanism can be found efficiently.  We also showed that optimal memoryless mechanisms are hard to compute and can be severely suboptimal relative to unconstrained mechanisms.  Our experimental results showed significant improvements relative to na\"ive approaches that do not take the agent's incentives into account, as well as that the difference in performance between facing a myopic and a patient agent is not large.  When the setting is sufficiently adversarial it is better to face a myopic agent, but otherwise it is better to face a patient agent.

Besides using these algorithms directly for appropriate applications, the experimental results that they enable (including those that we presented in this paper) can guide new theory.  For example, can we rigorously prove the benefit of facing a patient agent when the setting is not all too adversarial, and perhaps even characterize the transition point at which facing a patient agent becomes better than facing a myopic one?  Analytically derived mechanisms can also be compared to these experimental results to see how close to optimal in performance they are.  Finally, close inspection of the actual mechanisms generated by our algorithms may reveal insights that can be used to analytically design new mechanisms.

\bibliographystyle{plainnat}
\bibliography{ref}

\clearpage

\appendix
\section{Omitted Proofs from Section~\ref{sec:general}}

\begin{proof}[Proof of Theorem~\ref{thm:long}]
    We consider the case where payments are not allowed, i.e., $p_t(h, s) = 0$ for all $h \in \cH$ and $s \in \cS$.
    The case with payments and dynamic IR constraints is essentially the same.
    We use a similar reduction from MAX-SAT to the ones in \cite{papadimitriou1987complexity,mundhenk2000complexity} for partially observable Markov decision processes (POMDPs).
    Given a MAX-SAT instance with $n$ variables $x_1, \dots, x_n$ and $m$ clauses $c_1, \dots, c_m$ where $c_i = \{\ell_{i, j}\}_{j \in [k_i]}$ and each $\ell_{i, j}$ is a literal, we construct a dynamic environment where $T = n$, $|\cS| = m + 1$, and $|\cA| = 2$.
    The goal is to show that the maximum utility is precisely the fraction of clauses that can be simultaneously satisfied.
    Without loss of generality, we assume that no clause contains both the positive literal and the negative literal of a same variable.
    We first describe $\cS$ and $\cA$.
    Each clause $c_i$ corresponds to a unique state in $\cS$, $s_i$.
    In addition to these $m$ states, there is another state $s_0$.
    $\cA$ consists of two actions: $a_\mathrm{pos}$ and $a_\mathrm{neg}$.
    The transition operator and the principal's valuation function are such that:
    \begin{itemize}
        \item The initial distribution is uniform over $\{s_i\}_{i \in [m]}$, i.e., $\trans_0(s_i) = 1 / m$ for each $i \in [m]$.
        \item For each $t \in [T]$ and $a \in \cA$,
        \[
            \trans_t(s_0, a, s_0) = 1 \quad \text{and} \quad \vp_t(s_0, a) = 0.
        \]
        Moreover, for each $t \in [T]$ and $i \in [m]$:
        \begin{itemize}
            \item If $x_t^+ \in c_i$, then
            \[
                \trans_t(s_i, a_\mathrm{pos}, s_0) = \trans_t(s_i, a_\mathrm{neg}, s_i) = 1,
            \]
            and
            \[
                \vp_t(s_i, a_\mathrm{pos}) = 1 \quad \text{and} \quad \vp_t(s_i, a_\mathrm{neg}) = 0.
            \]
            \item if $x_t^- \in c_i$, then
            \[
                \trans_t(s_i, a_\mathrm{pos}, s_i) = \trans_t(s_i, a_\mathrm{neg}, s_0) = 1,
            \]
            and
            \[
                \vp_t(s_i, a_\mathrm{pos}) = 0 \quad \text{and} \quad \vp_t(s_i, a_\mathrm{neg}) = 1.
            \]
            \item otherwise,
            \[
                \trans_t(s_i, a_\mathrm{pos}, s_i) = \trans_t(s_i, a_\mathrm{neg}, s_i) = 1,
            \]
            and
            \[
                \vp_t(s_i, a_\mathrm{pos}) = \vp_t(s_i, a_\mathrm{neg}) = 0.
            \]
        \end{itemize}
        \item The principal and the agent are in a zero-sum situation, i.e., for any $t \in [T]$, $s \in \cS$, $a \in \cA$,
        \[
            \va_t(s, a) = 1 - \vp_t(s, a).
        \]
    \end{itemize}

    Now we show that the maximum utility is precisely the fraction of clauses that can be simultaneously satisfied.
    First observe that without loss of generality, an optimal mechanism depends only on time (and not on the reported states).
    This is because of the zero-sum situation: if the mechanism depends on the reports, then the agent can always choose the worst sequence of actions, which can only make the principal's utility smaller.
    Moreover, given the above observation, without loss of generality, an optimal mechanism is deterministic.
    This is because the overall utility of the principal is linear in the action at any time $t$, so one can always round a randomized mechanism into a deterministic one with at least the same overall utility.
    
    Given the above observations, an optimal mechanism corresponds precisely to a way of assigning values to variables in the MAX-SAT instance: for each $t \in [T]$, the action at time $t$ is $a_\mathrm{pos}$ iff the variable $x_t = 1$ (i.e., the literal $x_t^+$ is chosen).
    Moreover, when the initial state is $s_i$, the onward utility is $1$ if the clause $c_i$ is satisfied by the above assignment, and $0$ otherwise.
    Since the initial state is uniformly at random among $\{s_i\}_{i \in [m]}$, the maximum utility is precisely the maximum fraction of clauses that are satisfiable by some assignment.
    The theorem then follows from the fact that MAX-SAT is hard to approximate within a factor of $7/8 + \eps$ for any $\eps > 0$ \cite{haastad2001some}.
\end{proof}

\begin{proof}[Proof of Lemma~\ref{lem:obj}]
    For brevity, let $\mathrm{obj}$ denote the objective, i.e.,
    \[
        \mathrm{obj} = \sum_{h \in \cH, s \in \cS: (h, s)\text{ is feasible}} \left(\sum_{a \in \cA} \vp_{|h| + 1}(s, a) \cdot x(h, s, a) + y(h, s)\right).
    \]
    Moreover, for each $h \in \cH$, $s \in \cS$, let
    \[
        \mathrm{obj}(h, s) = \sum_{h' \in \cH, s' \in \cS: (h, s) \subseteq (h', s')} \left(\sum_{a \in \cA} \vp_{|h'| + 1}(s', a) \cdot x(h', s', a) + y(h', s')\right).
    \]
    Observe that
    \[
        \mathrm{obj} = \sum_{s \in \cS} \mathrm{obj}(\emptyset, s).
    \]

    We first prove inductively that for each $h \in \cH$, $s \in \cS$,
    \[
        \mathrm{obj}(h, s) = z(h, s) \cdot \up^M(h, s).
    \]
    When $|h| = T - 1$, by the definition of feasible extensions and the construction of the mechanism,
    \begin{align*}
        \mathrm{obj}(h, s) & = \sum_{a \in \cA} \vp_T(s, a) \cdot x(h, s, a) + y(h, s) \\
        & = z(h, s) \cdot \left(\sum_{a \in \cA} \vp_T(s, a) \cdot \pi(h, s, a) + p(h, s) \right) \\
        & = z(h, s) \cdot \up^M(h, s).
    \end{align*}
    When $|h| < T - 1$, for similar reasons,
    \begin{align*}
        \mathrm{obj}(h, s) & = \sum_{h', s': (h, s) \subseteq (h', s')} \left(\sum_{a \in \cA} \vp_{|h'| + 1}(s', a) \cdot x(h', s', a) + y(h', s')\right) \\
        & = \sum_{a \in \cA} \vp_{|h| + 1}(s, a) \cdot x(h, s, a) + y(h, s) \\
        & \phantom{\ =\ } + \sum_{h', s': (h, s) \subseteq (h', s'), |h'| > |h|} \left(\sum_{a \in \cA} \vp_{|h'| + 1}(s', a) \cdot x(h', s', a) + y(h', s')\right) \\
        & = z(h, s) \cdot \left(\sum_{a \in \cA} \vp_{|h| + 1}(s, a) \cdot \pi(h, s, a) + p(h, s)\right) \\
        & \phantom{\ =\ } + \sum_{a', s'': \trans_{|h| + 1}(s, a', s'') > 0} \sum_{h', s': (h + (s, a'), s'') \subseteq (h', s')} \left(\sum_{a \in \cA} \vp_{|h'| + 1}(s', a) \cdot x(h', s', a) + y(h', s')\right)
    \end{align*}
    By the induction hypothesis, the second sum above is equal to
    \begin{align*}
        & \phantom{\ =\ } \sum_{a', s'': \trans_{|h| + 1}(s, a', s'') > 0} \mathrm{obj}(h + (s, a'), s'') \\
        & = \sum_{a', s'': \trans_{|h| + 1}(s, a', s'') > 0} z(h + (s, a'), s'') \cdot \up^M(h + (s, a'), s'') \\
        & = \sum_{a', s'': \trans_{|h| + 1}(s, a', s'') > 0} x(h, s, a') \cdot \etrans_{|h| + 1}(s, a', s'') \cdot \up^M(h + (s, a'), s'') \\
        & = \sum_{a \in \cA, s' \in \cS} x(h, s, a) \cdot \trans_{|h| + 1}(s, a, s') \cdot \up^M(h + (s, a), s') \\
        & = z(h, s) \cdot \sum_{a \in \cA} \left(\pi(h, s, a) \cdot \sum_{s' \in \cS} \trans_{|h| + 1}(s, a, s') \cdot \up^M(h + (s, a), s')\right).
    \end{align*}
    Putting this back into the above expression for $\mathrm{obj}(h, s)$, we get
    \begin{align*}
        & \phantom{\ = \ } \mathrm{obj}(h, s) \\
        & = z(h, s) \cdot \left(\sum_{a \in \cA} \vp_{|h| + 1}(s, a) \cdot \pi(h, s, a) + p(h, s)\right) \\
        & \phantom{\ = \ } + z(h, s) \cdot \sum_{a \in \cA} \left(\pi(h, s, a) \cdot \sum_{s' \in \cS} \trans_{|h| + 1}(s, a, s') \cdot \up^M(h + (s, a), s')\right) \\
        & = z(h, s) \cdot \left(\sum_{a \in \cA} \cdot \pi_{|h| + 1}(h, s, a) \cdot \left(\vp_{|h| + 1}(s, a) + \sum_{s' \in \cS} \trans_{|h| + 1}(s, a, s') \cdot \up^M(h + (s, a), s')\right) + p(h, s)\right) \\
        & = z(h, s) \cdot \up^M(h, s).
    \end{align*}
    So for any $h \in \cH$, $s \in \cS$, $\mathrm{obj}(h, s) = z(h, s) \cdot \up^M(h, s)$.
    Then we immediately have
    \begin{align*}
        \up^M(\emptyset) & = \sum_{s \in \cS} \trans_0(s) \cdot \up^M(\emptyset, s) = \sum_{s \in \cS} z(\emptyset, s) \cdot \up^M(\emptyset, s) = \sum_{s \in \cS} \mathrm{obj}(\emptyset, s) = \mathrm{obj}. \qedhere
    \end{align*}
\end{proof}

\begin{proof}[Proof of Lemma~\ref{lem:ic1}]
    By Eq.~\eqref{eq:flow3} and Lemma~\ref{lem:util}, for all $h$, $s$, $s'$,
    \begin{align*}
        & \phantom{\ = \ } u(h, s, s') \\
        & = \sum_{a \in \cA} \va_{|h| + 1}(s, a) \cdot x(h, s', a) - y(h, s') \\
        & \phantom{\ =\ } + \sum_{a \in \cA, s'' \in \cS} \frac{\trans_{|h| + 1}(s, a, s'')}{\etrans_{|h| + 1}(s', a, s'')} \cdot z(h + (s', a), s'') \cdot \ua^M(h + (s', a), s'') \tag{Lemma~\ref{lem:util}} \\
        & = \sum_{a \in \cA} \va_{|h| + 1}(s, a) \cdot x(h, s', a) - y(h, s') + \sum_{a \in \cA, s'' \in \cS} \trans_{|h| + 1}(s, a, s'') \cdot x(h, s', a) \cdot \ua^M(h + (s', a), s'') \tag{Eq.~\eqref{eq:flow3}}
    \end{align*}
    Now by rearranging the above expression and applying the construction of the mechanism $M$ and the single-step reporting strategy $r_{h, s, s'}$, we have
    \begin{align*}
        & \phantom{\ = \ } u(h, s, s') \\
        & = \sum_{a \in \cA} x(h, s', a) \left(\va_{|h| + 1}(s, a) + \sum_{s'' \in \cS} \trans_{|h| + 1}(s, a, s'') \cdot \ua^M(h + (s', a), s'')\right) - y(h, s') \tag{rearranging} \\
        & = z(h, s') \cdot \left(\sum_{a} \pi(h, s', a) \cdot \left(\va_{|h| + 1}(s, a) + \sum_{s''} \trans_{|h| + 1}(s, a, s'') \cdot \ua^M(h + (s', a), s'')\right) - p(h, s')\right) \tag{construction of mechanism} \\
        & = z(h, s') \cdot \left(\sum_{a} \pi(h, s', a) \cdot \left(\va_{|h| + 1}(s, a) + \sum_{s''} \trans_{|h| + 1}(s, a, s'') \cdot \ua^{M, r_{h, s, s'}}(h + (s', a), s'')\right) - p(h, s')\right) \tag{construction of $r_{h, s, s'}$} \\
        & = z(h, s') \cdot \ua^{M, r_{h, s, s'}}(h, s), \tag{definition of $\ua^{M, r_{h, s, s'}}$}
    \end{align*}
    as desired.
\end{proof}

\begin{proof}[Proof of Lemma~\ref{lem:ic2}]
    Fix $h \in \cH$, $s, s' \in \cS$, and let $(s_p, a_p) = \last(h)$.
    When $h = \emptyset$, by Lemmas~\ref{lem:util}~and~\ref{lem:ic1} and Eq.~\eqref{eq:flow2},
    \begin{align*}
        & \phantom{\ \iff \ } u(h, s) \ge \frac{\etrans_{|h|}(s_p, a_p, s)}{\etrans_{|h|}(s_p, a_p, s')} \cdot u(h, s, s') \\
        & \iff z(\emptyset, s) \cdot \ua^M(\emptyset, s) \ge \frac{\etrans_{0}(s_p, a_p, s)}{\etrans_{0}(s_p, a_p, s')} \cdot z(\emptyset, s') \cdot \ua^{M, r_{\emptyset, s, s'}}(\emptyset, s) \\
        & \iff z(\emptyset, s) \cdot \ua^M(\emptyset, s) \ge \frac{\etrans_{0}(s)}{\etrans_{0}(s')} \cdot z(\emptyset, s') \cdot \ua^{M, r_{\emptyset, s, s'}}(\emptyset, s) \\
        & \iff \ua^M(\emptyset, s) \ge \ua^{M, r_{\emptyset, s, s'}}(\emptyset, s).
    \end{align*}
    When $|h| > 0$, suppose $h = (s_1, a_1, \dots, s_t, a_t)$, and let $h_p = (s_1, a_1, \dots, s_{t - 1}, a_{t - 1})$.
    By Lemmas~\ref{lem:util}~and~\ref{lem:ic1} and Eq.~\eqref{eq:flow1},
    \begin{align*}
        & \phantom{\ \iff \ } u(h, s) \ge \frac{\etrans_{|h|}(s_p, a_p, s)}{\etrans_{|h|}(s_p, a_p, s')} \cdot u(h, s, s') \\
        & \iff z(h, s) \cdot \ua^M(h, s) \ge \frac{\etrans_{|h|}(s_p, a_p, s)}{\etrans_{|h|}(s_p, a_p, s')} \cdot z(h, s') \cdot \ua^{M, r_{h, s, s'}}(h, s) \\
        & \iff x(h_p, s_p, a_p) \cdot \ua^M(h, s) \ge x(h_p, s_p, a_p) \cdot \ua^{M, r_{h, s, s'}}(h, s).
    \end{align*}
    Note that when $x(h_p, s_p, a_p) = 0$, $(h, s)$ cannot be reached, because (1) if $z(h_p, s_p) > 0$, then when the (reported) history-state pair is $(h_p, s_p)$, the mechanism never takes action $a_p$, and (2) if $z(h_p, s_p) = 0$, then such an impossible action exists somewhere in $h_p$.
    In such cases, $\pi(h, s)$ and $p(h, s)$ will never be accessed, since it is impossible for the (reported) history to be $h$.
    In other words, when $(h, s)$ is reachable, we must have $x(h_p, s_p, a_p) > 0$, in which case the last inequality is equivalent to $\ua^M(h, s) \ge \ua^{M, r_{h, s, s'}}(h, s)$.
\end{proof}

\begin{proof}[Proof of Lemma~\ref{lem:ic}]
    We only need to show that IC is equivalent to robustness against single-step misreporting.
    We prove this inductively, aiming to eliminate misreporting one step at a time.
    To be more specific, consider the following partial reporting strategy.
    For a reporting strategy $r$, $t \in [T]$, let $r|_{\ge t}$ denote the reporting strategy restricted to time $t, t + 1, \dots, T$, i.e., for any $h' \in \cH$, $s' \in \cS$,
    \[
        r|_{\ge t}(h', s') = \begin{cases}
            s', \text{if } |h'| + 1 < t \\
            r(h', s'), \text{otherwise}
        \end{cases}.
    \]
    Similarly, let $r|_{< t}$ denote $r$ restricted to time $1, 2, \dots, t - 1$, and $r|_{= t}$ denote $r$ restricted to time $t$.
    We show inductively that for any reachable history-state pair $(h, s)$, and any reporting strategy $r$,
    \[
        \ua^{M, (r|_{< |h| + 1})}(h, s) \ge \ua^{M, r}(h, s).
    \]
    Without loss of generality, we assume that for any unreachable pair $(h', s')$, $r$ simply reports truthfully, i.e., $r(h', s') = s'$.

    Recall that $r(h)$ is the reported history given by $r$ when the true history is $h$.
    When $|h| = T - 1$, the above claim is implied by Lemma~\ref{lem:ic2}, because
    \[
        \ua^{M, r}(h, s) = \ua^{M, (r|_{\ge T})}(r(h), s) \ge \ua^M(r(h), s) = \ua^{M, (r|_{< T})}(h, s).
    \]
    Now suppose $|h| < T - 1$.
    By the induction hypothesis, we have
    \[
        \ua^{M, r}(h, s) = \ua^{M, (r|_{\ge |h| + 1})}(r(h), s) \le \ua^{M, ((r|_{\ge |h| + 1})|_{< |h| + 2})}(r(h), s) = \ua^{M, (r|_{= |h| + 1})}(r(h), s).
    \]
    Now again by Lemma~\ref{lem:ic2}, we have
    \[
        \ua^{M, r}(h, s) \le \ua^{M, (r|_{= |h| + 1})}(r(h), s) \le \ua^M(r(h), s) = \ua^{M, (r|_{< |h| + 1})}(h, s),
    \]
    which establishes the above claim.

    Now observe that as a special case of the claim, for any $s \in \cS$,
    \[
        \ua^{M, r}(\emptyset, s) \le \ua^{M, (r|_{< 1})}(\emptyset, s) = \ua^M(\emptyset, s).
    \]
    Now summing over $s$, this implies that for any reporting strategy $r$,
    \[
        \ua^{M, r}(\emptyset) = \sum_{s \in \cS} \trans_0(s) \cdot \ua^{M, r}(\emptyset, s) \le \sum_{s \in \cS} \trans_0(s) \cdot \ua^M(\emptyset, s) = \ua^M(\emptyset). \qedhere
    \]
\end{proof}

\begin{proof}[Proof of Theorem~\ref{thm:short}]
    Given the correspondence between mechanisms and LP variables, by Lemma~\ref{lem:ic}, it is easy to see that (modulo the unreachable parts) every IC and IR mechanism corresponds bijectively to a feasible solution to the LP in Figure~\ref{fig:lp}.
    Moreover, by Lemma~\ref{lem:obj}, the objective value of this solution is precisely the principal's overall utility, which directly implies that an optimal solution to the LP corresponds to an IC and IR mechanism which maximizes the principal's overall utility.

    Now observe that the number of variables and the number of constraints in the LP are both $O(|\cS|^{T + 1} |\cA|^T)$.
    Moreover, all relevant coefficients in the LP can be encoded using $O(L)$ bits.
    It is well-known that such an LP can be solved in time $\mathrm{poly}(|\cS|^T, |\cA|^T, L)$.
\end{proof}

\section{Omitted Proofs from Section~\ref{sec:myopic}}

\begin{proof}[Proof of Lemma~\ref{lem:characterization}]
    We construct $M'$ explicitly based on $M$.
    Let $\pi'(t, s_p, a_p, s, a)$ be the probability that $M'$ chooses action $a$ at time $t$ in state $s$ when the previous state-action pair is $(s_p, a_p)$.
    Similarly, let $p'(t, s_p, a_p, s)$ be the payment specified by $M'$ at time $t$ in state $s$ when the previous state-action pair is $(s_p, a_p)$.
    We construct $M'$ from $M$ inductively as follows.
    For each $t \in [T]$, $s_p \in \cS$ and $a_p \in \cA$, let $h^*(t, s_p, a_p) \in \cH_{t - 1}$ be any history such that
    \begin{align*}
        h^*(t, s_p, a_p) \in \argmax_{h \in \cH_{t - 1}: (s_p, a_p) = \last(h)} \sum_{s \in \cS} \trans_{t - 1}(s_p, a_p, s) \cdot & \left(p(h, s) + \sum_{a \in \cA} \pi(h, s, a) \cdot \left(\vp_{|h| + 1}(s, a) \vphantom{\sum_{s' \in \cS}} \right. \right. \\
        & \left. \left. + \sum_{s' \in \cS} \trans_t(s, a, s') \cdot \up^M(h + (s, a), s')\right)\right).
    \end{align*}
    Then, for all $s \in \cS$, let
    \[
        \pi'(t, s_p, a_p, s) = \pi(h^*(t, s_p, a_p), s) \quad \text{and} \quad p(t, s_p, a_p, s) = p(h^*(t, s_p, a_p), s).
    \]
    This finishes the construction of $M'$.

    We first show that $\up^{M'}(\emptyset) \ge \up^M(\emptyset)$, by inductively showing a stronger claim: for all $h \in \cH$,
    \[
        \sum_s \trans_{|h|}(s_p, a_p, s) \cdot \up^{M'}(h, s) \ge \sum_s \trans_{|h|}(s_p, a_p, s) \cdot \up^M(h, s),
    \]
    where $(s_p, a_p) = \last(h)$.
    For all $h \in \cH_{T - 1}$, letting $(s_p, a_p) = \last(h)$, by the construction of $M'$, we have
    \begin{align*}
        \sum_s \trans_{T - 1}(s_p, a_p, s) \cdot \up^{M'}(h, s) & = \sum_s \trans_{T - 1}(s_p, a_p, s) \cdot \up^M(h^*(T, s_p, a_p), s) \\
        & \ge \sum_s \trans_{T - 1}(s_p, a_p, s) \cdot \up^M(h, s).
    \end{align*}
    Now for all $h \in \cH$ where $|h| < T - 1$, letting $(s_p, a_p) = \last(h)$ and $h^* = h^*(|h| + 1, s_p, a_p)$, we have
    \begin{align*}
        & \phantom{\ =\ } \sum_s \trans_{|h|}(s_p, a_p, s) \cdot \up^{M'}(h, s) \\
        & = \sum_s \trans_{|h|}(s_p, a_p, s) \cdot \left(p(h^*, s) + \sum_{a \in \cA} \pi(h^*, s, a) \cdot \left(\vp_{|h| + 1}(s, a) \vphantom{\sum_{s' \in \cS}} + \sum_{s' \in \cS} \trans_t(s, a, s') \cdot \up^{M'}(h + (s, a), s')\right)\right) \\
        & = \sum_s \trans_{|h|}(s_p, a_p, s) \cdot \left(p(h^*, s) + \sum_{a \in \cA} \pi(h^*, s, a) \cdot \left(\vp_{|h| + 1}(s, a) \vphantom{\sum_{s' \in \cS}} + \sum_{s' \in \cS} \trans_t(s, a, s') \cdot \up^{M'}(h^* + (s, a), s')\right)\right) \tag{property of $M'$} \\
        & \ge \sum_s \trans_{|h|}(s_p, a_p, s) \cdot \left(p(h^*, s) + \sum_{a \in \cA} \pi(h^*, s, a) \cdot \left(\vp_{|h| + 1}(s, a) \vphantom{\sum_{s' \in \cS}} + \sum_{s' \in \cS} \trans_t(s, a, s') \cdot \up^M(h^* + (s, a), s')\right)\right) \tag{induction hypothesis} \\
        & \ge \sum_s \trans_{|h|}(s_p, a_p, s) \cdot \left(p(h, s) + \sum_{a \in \cA} \pi(h, s, a) \cdot \left(\vp_{|h| + 1}(s, a) \vphantom{\sum_{s' \in \cS}} + \sum_{s' \in \cS} \trans_t(s, a, s') \cdot \up^M(h + (s, a), s')\right)\right) \tag{choice of $h^*$} \\
        & = \sum_s \trans_{|h|}(s_p, a_p, s) \cdot \up^M(h, s).
    \end{align*}
    Then in particular, we have
    \[
        \up^{M'}(\emptyset) = \sum_s \trans_0(s) \cdot \up^{M'}(\emptyset, s) \ge \sum_s \trans_0(s) \cdot \up^M(\emptyset, s) = \up^M(\emptyset).
    \]

    Finally we prove that $M'$ is IC.
    By the proof of Lemma~\ref{lem:ic}, we only need to show that $M'$ is robust against any single-step reporting strategy $r_{h, s, s'}$.
    In fact, letting $(s_p, a_p) = \last(h)$ and $h^* = h^*(|h| + 1, s_p, a_p)$,
    \[
        \ua^{M'}(h, s) = \sum_a \pi(h^*, s, a) \cdot \va_{|h| + 1}(s, a) + p(h^*, s) = \ua^M(h^*, s).
    \]
    Moreover,
    \[
        \ua^{M', r_{h, s, s'}}(h, s) = \sum_a \pi(h^*, s', a) \cdot \va_{|h| + 1}(s, a) + p(h^*, s) = \ua^{M, r_{h, s, s'}}(h^*, s).
    \]
    Since $M$ is IC, we have
    \[
        \ua^{M'}(h, s) = \ua^M(h^*, s) \ge \ua^{M, r_{h, s, s'}}(h^*, s) = \ua^{M', r_{h, s, s'}}(h, s).
    \]
    Now by the argument in the proof of Lemma~\ref{lem:ic}, we know that for all reporting strategy $r$, $h \in \cH$, $s \in \cS$,
    \[
        \ua^{M', r}(h, s) \le \ua^{M', (r|_{< |h| + 1})}(h, s),
    \]
    so
    \[
        \ua^{M', (r|_{\ge |h| + 1})}(h, s) \le \ua^{M', ((r|_{\ge |h| + 1})|_{< |h| + 1})}(h, s) = \ua^{M'}(h, s),
    \]
    which is precisely the IC requirement for myopic agents.
    Similar arguments guarantee that $M'$ has the same IR property as $M$.
\end{proof}

\begin{proof}[Proof of Theorem~\ref{thm:myopic}]
    We first argue the easy part, i.e., the time complexity.
    Observe that calls to $\optstatmech$ dominates the time complexity.
    Moreover, the algorithm makes $T |\cS| |\cA|$ calls to $\optstatmech$, so the overall time complexity is as stated.

    Now we show the optimality of the computed mechanism $M$.
    We prove inductively a stronger claim, i.e., for any $t \in [T]$, $s_p \in \cS$, $a_p \in \cA$,
    \[
        \sum_s \trans_0(s_p, a_p, s) \cdot \up^M(t, s_p, a_p, s) = \max_{M'} \trans_0(s_p, a_p, s) \cdot \up^{M'}(t, s_p, a_p, s),
    \]
    where the maximum is over all succinct mechanisms $M'$ that are IC and (optionally) IR.
    First observe that for all $s \in \cS$, $a \in \cA$,
    \[
        u(T, s, a) = \vp_T(s, a).
    \]
    So, for all $s_p \in \cS$, $a_p \in \cA$,
    \begin{align*}
        & \phantom{\ =\ } \sum_s \trans_0(s_p, a_p, s) \cdot \up^M(T, s_p, a_p, s) \\
        & = \sum_s \trans_0(s_p, a_p, s) \cdot \left(p(T, s_p, a_p, s) + \sum_a \pi(T, s_p, a_p, s, a) \cdot \vp_T(s, a)\right) \\
        & = \max_{M' = (\pi', p')} \sum_s \trans_0(s_p, a_p, s) \cdot \left(p'(T, s_p, a_p, s) + \sum_a \pi'(T, s_p, a_p, s, a) \cdot \vp_T(s, a)\right) \tag{optimality of $M$ at time $T$ as a static mechanism} \\
        & = \max_{M'} \sum_s \trans_0(s_p, a_p, s) \cdot \up^M(T, s_p, a_p, s).
    \end{align*}
    Again, the maximum is over all succinct mechanisms $M'$ that are IC and (optionally) IR.

    Now for $t \in [T - 1]$, by the construction of $M$,
    \begin{align*}
        & \phantom{\ =\ } \sum_s \trans_0(s_p, a_p, s) \cdot \up^M(t, s_p, a_p, s) \\
        & = \sum_s \trans_0(s_p, a_p, s) \cdot \left(p(t, s_p, a_p, s) + \sum_a \pi(t, s_p, a_p, s, a) \cdot \left(\vp_t(s, a) \vphantom{\sum_{s'}} \right.\right. \\
        & \phantom{\ =\ } \left.\left. + \sum_{s'} \trans_t(s, a, s') \cdot \up^M(t + 1, s, a, s')\right)\right) \\
        & = \max_{M' = (\pi', p')} \sum_s \trans_0(s_p, a_p, s) \cdot \left(p'(t, s_p, a_p, s) + \sum_a \pi'(t, s_p, a_p, s, a) \cdot \left(\vp_t(s, a) \vphantom{\sum_{s'}} \right.\right. \\
        & \phantom{\ =\ } \left.\left. + \sum_{s'} \trans_t(s, a, s') \cdot \up^M(t + 1, s, a, s')\right)\right). \tag{optimality of $M$ at time $t$ as a static mechanism}
    \end{align*}
    By the induction hypothesis and the fact that $M'$ is succinct,
    \begin{align*}
        & \phantom{\ =\ } \sum_s \trans_0(s_p, a_p, s) \cdot \up^M(t, s_p, a_p, s) \\
        & = \max_{M' = (\pi', p')} \sum_s \trans_0(s_p, a_p, s) \cdot \left(p'(t, s_p, a_p, s) + \sum_a \pi'(t, s_p, a_p, s, a) \cdot \left(\vp_t(s, a) \vphantom{\sum_{s'}} \right.\right. \\
        & \phantom{\ =\ } \left.\left. + \max_{M''} \sum_{s'} \trans_t(s, a, s') \cdot \up^{M''}(t + 1, s, a, s')\right)\right) \tag{induction hypothesis} \\
        & = \max_{M' = (\pi', p')} \sum_s \trans_0(s_p, a_p, s) \cdot \left(p'(t, s_p, a_p, s) + \sum_a \pi'(t, s_p, a_p, s, a) \cdot \left(\vp_t(s, a) \vphantom{\sum_{s'}} \right.\right. \\
        & \phantom{\ =\ } \left.\left. + \sum_{s'} \trans_t(s, a, s') \cdot \up^{M'}(t + 1, s, a, s')\right)\right) \tag{$M'$ is succinct} \\
        & = \max_{M'} \sum_s \trans_0(s_p, a_p, s) \cdot \up^{M'}(t, s_p, a_p, s).
    \end{align*}
    All maxima are over all succinct mechanisms that are IC and (optionally) IR.
    As a result, we have
    \[
        \up^M(\emptyset) = \sum_s \trans_0(s) \cdot \up^M(\emptyset, s) = \max_{M'} \sum_s \trans_0(s) \cdot \up^{M'}(\emptyset, s) = \max_{M'} \up^{M'}(\emptyset). \qedhere
    \]
\end{proof}

\section{Omitted Proofs from Section~\ref{sec:memoryless}}

\begin{proof}[Proof of Theorem~\ref{thm:suboptimality}]
    First suppose the agent is patient and without loss of generality has a discount factor of $1$.
    Let $T = 2$ and $\cS = \cA = [n]$ where $n \ge \eps^{-1}$.
    The initial distribution is uniform over $[n]$, i.e., $\trans_0(i) = 1 / n$ for all $i \in [n]$, i.e., no matter what action is played, all states always transition to state $1$.
    The transition operator is such that $\trans_1(i, j, 1) = 1$ for all $i, j \in [n]$.
    At time $T = 2$, the principal's valuations are $\vp_T(i, j) = 0$ for all $i, j \in [n]$.
    At time $1$, the principal's valuation function is such that for all $i, j \in [n]$, $\vp_1(i, j) = 1$ if $i = j$, and $\vp_1(i, j) = 0$ if $i \ne j$.
    For $t \in [T]$, the agent's valuation function is such that for all $i, j \in [n]$, $\va_t(i, j) = 0$ if $i = j$, and $\va_t(i, j) = 1$ if $i \ne j$.

    Consider the principal's optimal utility, which is clearly upper bounded by $1$ ($1$ at time $1$ and $0$ at time $2$).
    The following mechanism is IC and achieves this upper bound:
    \begin{itemize}
        \item At time $1$, play action $i$ for each state $i \in [n]$.
        \item At time $T = 2$, play action $(i \bmod n) + 1$ iff the state at time $1$ is $i$.
    \end{itemize}
    The mechanism is IC because regardless of the (reported) initial state, the agent achieves overall utility $1$.
    It is easy to check this mechanism achieves utility $1$.

    On the other hand, any memoryless IC mechanism can achieve utility at most $1 / n \le \eps$.
    This is because at time $T = 2$, the current state provides absolutely no information, so the mechanism has to perform the same (randomized) action regardless of the initial state.
    As a result, in order to be IC, the mechanism has to satisfy the following condition at time $1$: for all $i, j \in [n]$, $\pi(i, i) \le \pi(j, i)$, where $\pi(a, b)$ is the probability that action $b$ is played in state $a$ at time $1$.
    So the principal's utility can be bounded as follows:
    \[
        \frac1n \sum_i \pi(i, i) \le \frac1n \sum_i \left(\frac1n \sum_j \pi(j, i)\right) = \frac{1}{n^2} \sum_{i, j} \pi(j, i) = \frac1n.
    \]
    This concludes the proof when the agent is patient.

    Now consider the case with a myopic agent.
    Again, let $T = 2$ and $\cS = \cA = [n]$ where $n \ge \eps^{-1}$.
    The initial distribution is again uniform over $[n]$, i.e., $\trans_0(i) = 1 / n$ for all $i \in [n]$.
    The transition operator is such that $\trans_1(i, j, i) = 1$ for all $i, j \in [n]$, i.e., no matter what action is played, state $i$ always transitions to state $i$.
    At time $1$, the principal's and the agent's valuations are $\vp_1(i, j) = \va_1(i, j) = 0$ for all $i, j \in [n]$.
    At time $T = 2$, the principal's valuation function is such that for all $i, j \in [n]$, $\vp_T(i, j) = 1$ if $i = j$, and $\vp_T(i, j) = 0$ if $i \ne j$.
    And the agent's valuation function is such that for all $i, j \in [n]$, $\va_T(i, j) = 0$ if $i = j$, and $\va_T(i, j) = 1$ if $i \ne j$.

    The principal's optimal utility, $1$, is achieved by the following succinct (but not memoryless) IC mechanism:
    \begin{itemize}
        \item At time $1$, play action $1$ for all states.
        \item At time $2$, play action $i$ iff the state at time $1$ is $i$.
    \end{itemize}
    The mechanism is IC in particular because the agent is myopic and cannot change the past.
    It is easy to check this mechanism achieves utility $1$.

    On the other hand, any memoryless IC mechanism can achieve utility at most $1 / n \le \eps$.
    This is because at time $T = 2$, the mechanism cannot memorize anything before, so it has to be IC based only on the current state, which puts the mechanism in a situation that is essentially the same as at time $1$ in the hard instance for patient agents.
    Similar arguments then guarantee that the principal's utility is at most $1 / n$, which concludes the proof for myopic agents.
    Finally, we note that the above constructions work even if payments are allowed.
\end{proof}

\begin{proof}[Proof of Theorem~\ref{thm:np-hardness}]
    We use reductions from MAX-SAT similar to that in Theorem~\ref{thm:long} for both myopic and patient agents.
    First consider the case where the agent is patient with a discount factor of $1$.
    In this case, the reduction in Theorem~\ref{thm:long} applies without any modification.
    In particular, since the principal and the agent are in a zero-sum situation, without loss of generality, any optimal memoryless mechanism does not depend on the reported states.
    And again, since the principal's utility is multilinear in the actions at each time, there is a deterministic mechanism which is optimal.
    As argued in the proof of Theorem~\ref{thm:long}, such a mechanism corresponds precisely to an optimal assignment of variables in the MAX-SAT instance, which implies the $7/8 + \eps$ inapproximability.

    Now consider the case where the agent is myopic.
    Here we slightly modify the reduction, and in particular, the agent's valuation functions.
    That is, for each $t \in [T]$ and $i \in [m]$, we let
    \[
        \va_t(s, a_\mathrm{pos}) = c \quad \text{and} \quad \va_t(s, a_\mathrm{neg}) = 0,
    \]
    for all $s \in \cS$, where $c > 0$ is an arbitrarily small constant.
    This guarantees that at any time $t$, in order to be IC, the (randomized) actions for all states have to be exactly the same.
    Then since the principal's utility is multilinear, again it is without loss of generality to consider deterministic mechanisms, which correspond to assignments of variables.
    The ratio of $7/8 + \eps$ follows immediately.
    Finally, we remark that the above reductions still work when payments are allowed.
\end{proof}

\end{document}